\documentclass[lettersize,journal]{IEEEtran}
\usepackage{amsmath,amsfonts}
\usepackage{array}
\usepackage{subfig}
\usepackage{textcomp}
\usepackage{stfloats}
\usepackage{url}
\usepackage{verbatim}
\usepackage{graphicx}
\usepackage{cite}
\usepackage{booktabs}
\usepackage{subfloat}
\usepackage[ruled, vlined]{algorithm2e}
\newtheorem{theorem}{Theorem}
\newtheorem{proof}{Proof}
\newtheorem{definition}{\bf Definition} 
\begin{document}

\title{Federated Learning based on Defending Against Data Poisoning Attacks in IoT}

\author{Jiayin Li, Wenzhong Guo,~\IEEEmembership{Member,~IEEE}, Xingshuo Han, Jianping Cai, Ximeng Liu*,~\IEEEmembership{Senior Member,~IEEE}
\thanks{J. Li is with the College of Computer and Cyber Security, Fujian Normal University, Fuzhou, Fujian, China, 350117. W. Guo, J. Cai, X. Liu are with College of Computer and Data Science, Fuzhou University, Fuzhou, China, 350108; and Key Lab of Information Security of Network Systems (Fuzhou University), Fujian Province. X. Han is with School of Computer Science and Engineering, Nanyang Technological University, Singapore (email: lijiayin2019@gmail.com, guowenzhong@fzu.edu.cn, xingshuo.han@ntu.edu.sg, jpingcai@163.com, snbnix@gmail.com)}
\thanks{Manuscript received xxx xx, xxxx; revised xxx xx, xxxx.}
}

\markboth{Journal of \LaTeX\ Class Files,~Vol.~XX, No.~XX, XX~XX}%
{Shell \MakeLowercase{\textit{et al.}}: A Sample Article Using Internet of Things for IEEE Journals}


\maketitle

\begin{abstract}
The rapidly expanding number of Internet of Things (IoT) devices is generating huge quantities of data, but the data privacy and security exposure in IoT devices, especially in the automatic driving system. Federated learning (FL) is a paradigm that addresses data privacy, security, access rights, and access to heterogeneous message issues by integrating a global model based on distributed nodes. However, data poisoning attacks on FL can undermine the benefits, destroying the global model's availability and disrupting the model training. To avoid the above issues, we build up a hierarchical defense data poisoning (HDDP) system framework to defend against data poisoning attacks in FL, which monitors each local model of individual nodes via abnormal detection to remove the malicious clients. Whether the poisoning defense server has a trusted test dataset, we design the \underline{l}ocal \underline{m}odel \underline{t}est \underline{v}oting (LMTV) and \underline{k}ullback-\underline{l}eibler divergence \underline{a}nomaly parameters \underline{d}etection (KLAD) algorithms to defend against label-flipping poisoning attacks. Specifically, the trusted test dataset is utilized to obtain the evaluation results for each classification to recognize the malicious clients in LMTV. More importantly, we adopt the kullback leibler divergence to measure the similarity between local models without the trusted test dataset in KLAD. Finally, through extensive evaluations and against the various label-flipping poisoning attacks, LMTV and KLAD algorithms could achieve the $100\%$ and $40\%$ to $85\%$ successful defense ratios under different detection situations. 
\end{abstract}


\section{Introduction}
\IEEEPARstart{A}{s} represented autonomous vehicles Internet of Things (IoT) device number has increased dramatically, deep learning techniques have brought a paradigm in the variable application field, such as driven application \cite{ref1}, speech recognition \cite{ref2}, and image identification \cite{ref3}, which can use a large of history data to mine hidden patterns and make predictions about what will not happen in the future. As a result, people can optimize their travel time using deep learning techniques, significantly improving their work efficiency. Since traditional deep learning is a data-driven technique, the core of its performance is the size and quality of the training data, which is to train a usable network model through deep learning algorithms \cite{ref4, ref5} via collecting a mass of training dataset to store in a central server. Meanwhile, the privacy and security of data have been significantly compromised to bring a significant threat to people's lives and property \cite{ref6, ref7, ref8} and violate the general data protection regulation (GDPR) legal provisions \cite{ref9}. To address the above challenge, the federated learning (FL) approach is built up to avoid the problem of data privacy leakage and solve the disadvantages of training datasets islands \cite{ref10, ref11, ref12}. Unlike traditional deep learning framework, the FL is similar to a distributed learning structure that greatly solves the drawbacks of small training datasets caused by data decentralization. It can better use various data resources to train a more accurate neural network prediction model. In addition, the FL breaks the barrier that different data owners can not share their private data without worrying about data leakage. Despite the advantages mentioned above, the FL increases the risk of malicious attacks \cite{ref13, ref14, ref15}. A typical FL structure contains many local clients and a central server defined as the trusted parties. Due to the vast number of parties involved, each participant is not guaranteed to be an honest and credible entity, suffering from data poisoning attacks to reduce the prediction accuracy of the final model seriously \cite{ref16, ref17, ref18}.

The goal of data poisoning attacks is to attack the neural network model by breaking the distribution features of the training datasets, which are pervasive means for neural network vicious backdoor implantation to reach a specific purpose based on the attacker's intent  \cite{ref19, ref20, ref21}. When the training dataset suffers from data poisoning attacks, the trained network model based on the poisoned dataset can reduce the classification accuracy and lead to misclassified. Therefore, the above issues have laid an enormous security risk for applying neural networks, especially in the autonomous vehicles (AV) field, which leads to traffic accidents \cite{ref22}. Similarly, the data poisoning attacks are commonly used to destroy the final global model by attacking many local clients involved in FL \cite{ref23}. However, the existing defense approaches against data poisoning attacks are fundamental for building on the credible test dataset, and their defense success rate needs further improvement. Besides, it focuses on defending against data poisoning attacks without a trusted test dataset.

To guarantee the usability of the global predictive model and avoid data poisoning attacks in IoT, we propose a hierarchical defense data poisoning (HDDP) system against data poisoning attacks. Due to the local models' collaborative integration approach to obtaining the training model, the central server cannot receive the training dataset of each client to analyze the dataset. Thereby, the traditional approaches based on analyzing datasets were not suitable to reach the goal of finding poisoned local models in FL \cite{ref24, ref25, ref26}. To extend the usefulness of our proposed HDDP for FL, we build up the HDDP system to defend against data poisoning attacks under different detection conditions in this paper. The significant contributions of this paper are summarized as follows.
\begin{itemize}
\item In this paper, we propose a hierarchical defense HDDP against data poisoning attacks in FL, which can contain two types of defending against the data poisoning attacks to remove the poisoned local models and decrease the poisoned local models to influence the performance of an aggregated global model.  
\item To improve the performance of the defense mechanism against data poisoning attacks, we design a local model test voting (LMTV) algorithm based on the trusted test dataset to remove the poisoned local models. Specifically, we statistics the test classification results for each client's local model for each round of interaction in FL, which can be regarded as the basis for voting. 
\item To solve defense against data poisoning attacks in the absence of a trusted validation dataset, we also design a Kullback-Leibler (KL) divergence anomaly parameters detection (KLAD) algorithm, which computes the KL divergence to remove the irrelevant poisoned local models without the trusted test dataset. 
\item Finally, we implement the label-flipping attacks for data poisoning and evaluate the proposed defense algorithms via images classification based on Fashion-MINIST and GTSRB datasets to validate the efficiency and security of HDDP in FL, which reaches excellent performance.
\end{itemize}

The rest of the paper is organized as follows. In Section II, we give the related work. In Section III, we describe the preliminary knowledge. Furthermore, the system architecture is defined in Section IV. Next, we give more detail about the defense algorithms in Section V. Section VI provides the security and performance analysis. We conclude the paper and provide future work in Section VII.

\section{Related Work}
Data poisoning attacks are highly relevant for various application fields in IoT, which pose a severe security threat, such as computer vision \cite{ref27}, spam filtering \cite{ref28, ref29}, and disease diagnosis \cite{ref30}. In the beginning, several data poisoning attacks were studied to destroy some simple machine learning (ML) models containing support vector machine (SVM) \cite{ref31, ref32, ref33}, linear regression \cite{ref34}, and unsupervised learning \cite{ref35}, where data poisoning attacks researches could destroy the accuracy of the trained model by inserting some error data into original training data to disturb the distribution of training data. However, those data poisoning attacks only cause classification errors and do not achieve a specific target of poisoning attacks. To achieve purposeful data poisoning attacks, some target data poisoning attacks were proposed based on the will of malicious attackers \cite{ref36, ref37, ref38}. For instance, the authors in \cite{ref36} proposed a targeted backdoor attack to destroy the deep learning systems based on data poisoning, which realizes backdoor attacks by utilizing poisoning strategies (e.g., burying a specific back door) to achieve the adversarial goal of malicious attackers.
Moreover, the backdoor key is hard to notice even by humans to achieve stealthiness. However, the target backdoor attacks approach is less robust. To improve the robustness of data poisoning attacks, Wang $et$ $al.$ \cite{ref37} realized the data poisoning attacks against online learning to utilize a class of adversarial attacks on ML where an adversary has the power to alter a small fraction of the training data to make the trained classifier satisfy specific objectives. Nevertheless, it needs to have the ability to manipulate training data. Therefore, Zhao $et$ $al.$ \cite{ref38} proposed using a universal adversarial trigger as the backdoor trigger to attack video recognition models, which is an identification error directly based on the attacker's intent. However, the above studies for data poisoning attacks are often limited to centralized neural network training, which does not work well in FL. To improve the success rate of data poisoning attacks in FL, Fang $et$ $al.$ \cite{ref17, ref23} proposed a local model of data poisoning attacks to destroy the FL, which assumes an attacker has compromised some client devices. During the learning process, the attacker manipulates the local models' parameters on the compromised client devices. Therefore, they could realize the label-flipping for the training dataset to impact some local clients in FL.

For defense approaches against data poisoning attacks, the existing defense approaches aim to sanitize the training dataset. One category of defenses was to analyze the distribution of the training dataset and detect the malicious dataset based on their negative impact on the error of obtained model \cite{ref39, ref40, ref41}. Another category of defense approaches was to use the loss functions, optimizing the obtained model parameter to detect the injected malicious data simultaneously, which needs to jointly find a subset of the training dataset with a given size \cite{ref42, ref43}. Meanwhile, Peri $et$ $al.$ proposed a novel method by cleaning the error label to defend the data poisoning based on the deep-knn\cite{ref44}. Additionally, Chen $et$ $al.$ effectively defended the data poisoning attack in machine learning by adopting the GAN network \cite{ref45}. However, these defense approaches are not directly applicable to defense against the data poisoning attacks in FL because of the un-share of the local training datasets. For the defense algorithms in FL, there is also some literature to detect poisoning local models \cite{ref46, ref47}. Whatever these local model poisoning detection approaches that were implemented needed to be established based on the trusted test dataset. In addition, privacy breaches are risky during data poisoning detection. Designing the defense algorithm against data poisoning attacks is essential based on a credible verification dataset.

\section{Preliminary}
In this section, we introduce the basic knowledge to illustrate the detail of the data poisoning defense mechanism such as federated learning and Paillier cryptography.
\subsection{Federated Learning} 
Unlike traditional centralized deep learning framework, federated learning (FL) was proposed in the literature \cite{ref48} as a new training structure, which could ensure the privacy of private data and solve the data island. To call multiple resources, the FL has two substructures: many local clients and a central server. The function of local clients is to train local models based on their dataset without sharing the dataset. After each local client generates the local model, the central server can receive each local model to integrate the global model. Thereby, many local clients' interaction with the central server realizes the core of FL. To achieve the global model, the central server adopts the simple parameter average approach named FedAvg algorithm \cite{ref49}. More details about FL can be expressed as follows.

I. The central server initializes the training model $GM_0$, which is sent to each participant client $i\;(0<i\le n)$.

II. When each client receives the initialized model $GM_0$, the many local clients will update the $GM_0$ to generate a local model $LM_{(i, j)}$, where $j\;(0<j\le m)$ is the number of rounds of interactive communication between the local clients and the central server.

III. After each round, the central server could receive the local model $LM_{(i, j)}$ and further obtain the global model as the following formula.
\begin{equation}
	GM_j = \frac{\sum_{i = 1}^n{LM_{(i, j)}}}{n}
\end{equation}
Here, the end of model training sign is that the loss function reaches convergence.
\subsection{Paillier Cryptography}
Paillier cryptography is an effective technology to achieve additional homomorphic properties on the ciphertexts widely used in various privacy-preserving applications, which has three steps to realize data protection: key generation, message encryption, and ciphertexts decryption.
\begin{itemize}
\item Key Generation: For a security parameter $e$, select two large $e-$bit primes $p_1$, $q_1$, and compute the product $N$ of $p_1$, $q_1$ such as $N = p_1\cdot q_1$. Besides, the $p_1$ and $q_1$ should satisfy the least common multiple relationship $\lambda = lcm(p_1-1, q_1-1)$. Then, the function $L(x)$ is defined as $L(x) = \frac{x-1}{N}$, and calculate $\mu = (L(g^{\lambda}\mod N^2))^{-1}\mod N$, where $g\in \mathbb{Z}^{*}_n$. Finally, the public key and private key are $pk = (N, g)$ and $sk = (\lambda, \mu)$, respectively.
\item Message Encryption: Given a plaintext $m\in \mathbb{Z}_n$, after choosing a random value $r\in \mathbb{Z}^{*}_n$, the plaintext message $m$ can be encrypted as $c = Enc(m) ={ g^{m}\cdot r^{N}}\mod N^2$.
\item Ciphertext Decryption: Given a ciphertext $c$, the plaintext message is obtained by computing $m = Dec(c) = L(c^{\lambda})\cdot\mu\mod N$.
\end{itemize}
Note that Paillier cryptography has good additive homomorphism characteristics. Suppose that there are two plaintext $m_1\in \mathbb{Z}_n$ and $m_2 \in \mathbb{Z}_n$. According to the additive homomorphism characteristics, we can get the following formula.
\begin{equation}
Enc(m_1)\times Enc(m_2) = Enc(m_1+m_2)
\end{equation}
\subsection{Data Poisoning Attacks}
For the honest-but-curious (HBC) parties in FL, they passively observe the other data gathered from the clients to learn private messages about clients' datasets. Therefore, the FL is constantly being attacked by data poisoning. Unlike several known adversarial learning techniques reducing the accuracy of the predictive model, the data poisoning could trigger the classification error and accuracy drop to cause the unavailability of the training model based on a training dataset. For instance, the attackers can modify the test data to avoid the detection of malicious samples on a well-trained model \cite{ref50, ref51, ref52, ref53, ref54, ref55}. In this paper, we focus on target attacks and un-target attacks for data poisoning, which utilizes the training data to affect the global model accuracy and classification performance during the learning phase \cite{ref56}. Assume that the function $\mathcal{F}$ is the learning algorithm to train the predictive model. The $f$ is the data poisoning attacks algorithm if the malicious local clients obtain the poisoned dataset $\mathcal{D}^{*}$ from the benign dataset $\mathcal{D}$, they only execute the function $f$ is the following formula.
\begin{equation}
	f(\mathcal{D})\longrightarrow \mathcal{D}^{*}
\end{equation}

To realize the function $f$, the traditional approaches generate the error label to lead the error characteristics. Significantly, the poisoned training dataset $\mathcal{D}^{*}$ contains incorrect classification labels, which is very easy to distinguish the feature of the object $a$ into the feature of the object $b$. Thereby, the poisoned predictive model $GM^{*}$ can be obtained based on the $\mathcal{D}^{*}$ and learning algorithm $\mathcal{F}$, such as.
\begin{equation}
	\mathcal{D}^{*}\stackrel{\mathcal{F}}{\longrightarrow} GM^{*}
\end{equation}
Where the $GM^{*}$ will not achieve the correct classification of particular specific objects.

\section{System Architecture}
We aim to defend the impact of data poisoning attacks in FL and guarantee clients' privacy. Here, we build up the system model, define the attack model and further give the system's goal designed in this section.
\subsection{System Model}
In favor of the designed goal in this paper, we build up the system framework and solve the data poisoning defense mechanism to avoid the data poisoning attacks for object detection in IoT, which contains the vehicular clients training, data poisoning attacks, poisoning defense mechanism, and secure model aggregation as the Fig. 1.
\begin{figure}[!htb]
\centering
\includegraphics[width=80mm]{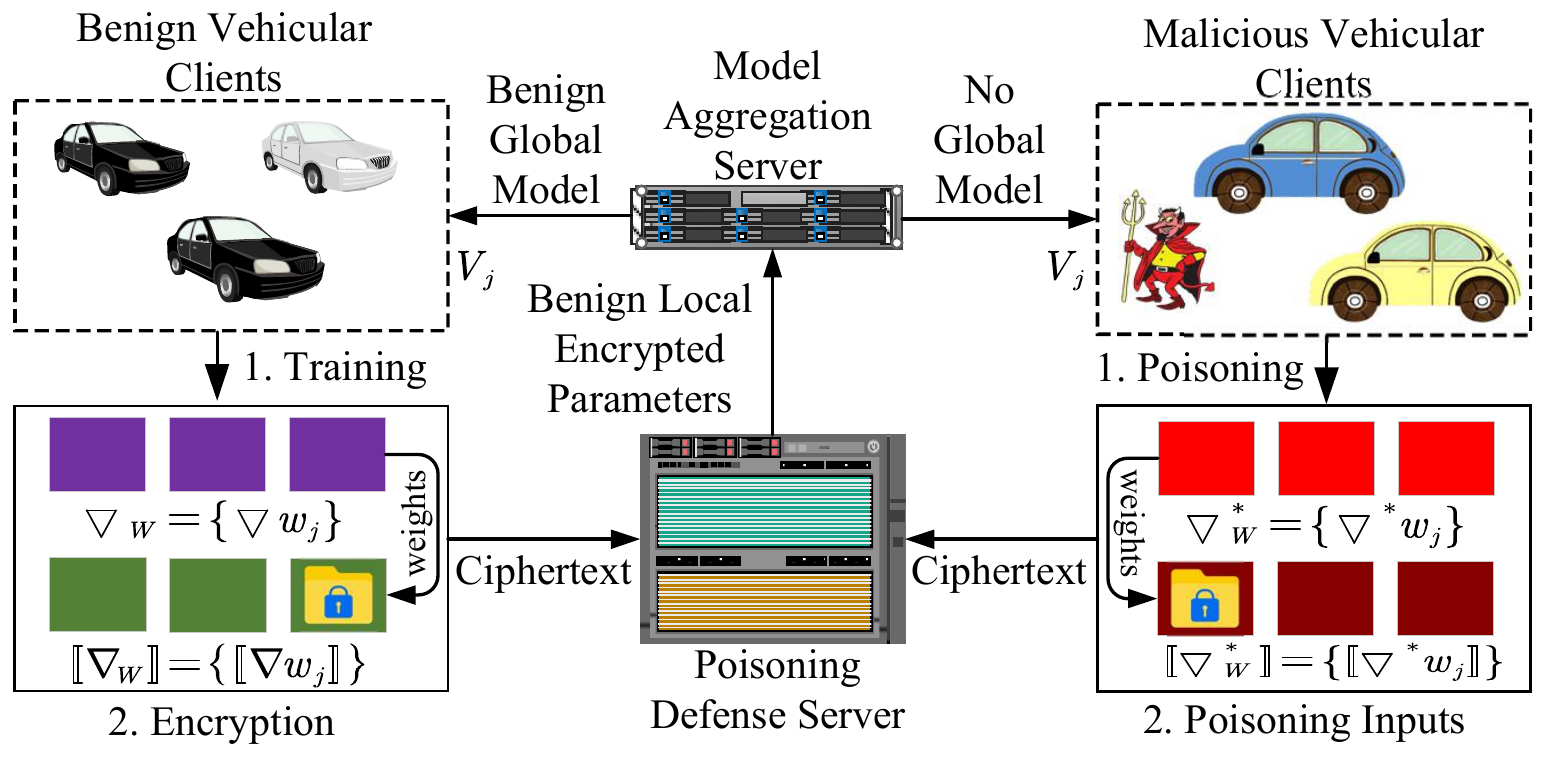}
\caption{System Model.}
\vspace{-0.1cm}
\end{figure}
\begin{itemize}
\item \textbf{Vehicular Client Trainers (VCT).} Vehicular client training is the first component to build up the local model by utilizing their dataset to train the model in the system model. And then, the vehicular clients (VCs) send the trained local model to the defense poisoning server. However, due to the vast number of participating VCs based on local models aggregation in FL, the system can not guarantee that each local vehicular client can honestly submit the benign model to the central server. On the contrary, malicious VCs could generate the poisoned local models to decrease the nature of the aggregated global model. Similarly, the VCs may be at risk of being attacked by malicious attackers without their knowledge providing the poisoned local models.
\item \textbf{Data Poisoning Attackers (DPA).} Data poisoning attacks are a common means to influence the performance-trained model in practical applications. This paper focuses on the label-flipping attacks in data poisoning to influence the FL. Furthermore, the label-flipping attacks are carried out to control and modify their vehicular dataset by malicious VCs or attackers. Malicious VCs and attackers aim to make the trained models achieve a specific classification for an error object according to their desired goal. Meanwhile, the poisoned models obtained are also sent and joined in global model integration.
\item \textbf{Poisoning Defense Server (PDS).} Poisoning defense server is an authenticated and a trusted party whose core function is designed in the system model to eliminate VCs' submitted poisoned local models. According to the algorithms we designed, the local models submitted must be detected to determine whether the local models are attacked before parameters collaborative integration. After detecting the data poisoning defense algorithms, the benign local models are saved to integrate with the usable global model.  
\item \textbf{Model Aggregation Server (MAS)} The MAS is also an HBC party, curious about the VCs' privacy but can perform the mechanism designed precisely as we want. When the poisoned local models are eliminated, the remaining benign ones are aggregated to generate a global model without decrypting the ciphertext of local models' parameters in MAS. Finally, the usable global model is sent to the benign VCs.
\end{itemize}
\subsection{Attack Model}
For data poisoning attacks in FL, we consider the constrained condition $\tau$ for the participated clients. The $\tau$ means the number of malicious clients, which usually need to satisfy $0<\tau<\frac{n}{2}$ to avoid detection easily for data poisoning attacks \cite{ref46}. Note that $n$ is the total amount of participants' clients. There are usually two different aspects of data poisoning to destroy the predicted ability of generated globe model in collaborative learning. One is called malicious clients that part some of the participating clients are malicious, which could control their dataset to reach the effect of destroying their model to impact the globe model. Another is called malicious attackers $\mathcal{A}$ that the participating clients may be attacked by the malicious attackers, which also modified the training dataset and triggered the data poisoning. To further state the attack process, we assume that all the adversaries know the learning algorithm used by each participating client. We can tamper with the training data accordingly (e.g., by adding fake data or mislabeling the training data). However, all the malicious users cannot collude together to poison their data and learn the globe model be-forehead as the malicious users know the training data generated from benign users.
\subsection{Problem Definition}
To illustrate the function of label-flipping attacks, we definite two indexes to measure the deep of label-flipping attacks, such as classification error and accuracy reduction.
\begin{definition}
\textbf{(Classification Error)} In data poisoning attacks, the attackers gave the source data $\mathcal{S}_d$, the global model $GM^{*}$ could output the target data $\mathcal{T}_{d}$ according to the attacker's wish, such as $\mathcal{S}_{d}\rightarrow \mathcal{T}_{d}$. Here, the $\mathcal{S}_d$ is an input value to the global model for which the attacker wants to influence the model's output. The $\mathcal{T}_{d}$ is a value of the attacker's choice that the influenced model should output for a given source input. The poisoning precision rate $Pr$ is defined as the following formula to measure the degree of data poisoning.
	\begin{equation}
		Pr = \frac{num(\zeta\!\rightarrow\!\varsigma)+num(\varsigma\!\rightarrow \!\zeta)}{num(\zeta)+num(\zeta\!\rightarrow \!\varsigma)+num(\varsigma)+num(\varsigma\!\rightarrow \!\zeta)}
	\end{equation}
where $\zeta$ and $\varsigma$ are the different classifications; $num(\zeta\rightarrow\varsigma)$ and $num(\varsigma\rightarrow\zeta)$ are the $\mathcal{S}_{d}\rightarrow \mathcal{T}_{d}$, and $num(\varsigma)$ are the number of right classification.
\end{definition}
\begin{definition}
\textbf{(Accuracy Reduction)} The accuracy reduction (AR) means that the data poisoning can also reduce the performance of the global model. Suppose that we define the $\rho$ to represent the classification accuracy. And then, the accuracy of classification for the benign global model $GM$ and the poisoned global model $GM^{*}$ are $\rho(GM)$ and $\rho(GM^{*})$. The following formula can express the specific quantification of AR.
	\begin{equation}
		AR_{(GM~ {GM^{*}})} = \rho(GM) - \rho(GM^{*})
	\end{equation}
\end{definition}
\subsection{Design Goal}
To defend against the impact of data poisoning (e.g., target attack and random poisoning attack), the system design aims to ensure classification correctness, defense success rate, and prediction model accuracy. 
\begin{itemize}
\item\textbf{Classification Correctness Rate:} To avoid the harm of data poisoning attacks, the system model can eliminate the malicious attackers' attacks on specific targets. More importantly, the system model should also be able to reject poisoned local clients from resubmitting models.
\item\textbf{Defense Success Rate:} The data poisoning defense success rate is another more important index to measure the efficiency of the poisoning defense algorithm, which is about the performance required to prevent the global model from the data poisoning attacks of local models from generating a practicable global model in FL, especially in the IoT.
\item\textbf{Prediction Model Accuracy:} For random data poisoning attacks, the higher the number of poisoned VCs, the lower the accuracy of the obtained global model prediction.  The system model must detect the malicious local models to reduce the impact of poisoned local clients' models submitted for the global model.
\end{itemize}

\section{Defense Mechanism in Federated Learning}
In this section, we mainly detail the hierarchical defense data poisoning model to eliminate the submitted local poisoned models. We give the following aspects: a system overview and the details about building the system.
\subsection{System Overview}
To facilitate the introduction of the system, we show a system overview. It contains four parts: training dataset generation, local models training, poisoning defense mechanism, and secure model aggregation, whose structure level is similar to an erect triangle. A brief overview of these stages is given in Fig. 2.
\begin{figure}[!htb]
\centering
\includegraphics[width=82mm]{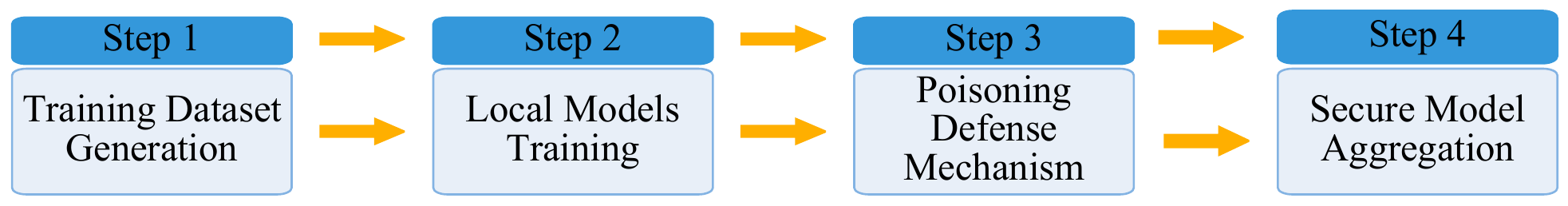}
\caption{System Overview.}
\end{figure}
\begin{itemize}
\item \textbf{Training Dataset Generation:} From Fig. 2, the training dataset generation is the basis for building up by utilizing the traditional FL structure. They are generated by the participating vehicles in the IVN, which also are stored in their own storage space to train the corresponding predictive models. Suppose that there are $n$ vehicles participating in collaborative learning training to generate the global model together. Each vehicle $V_{i\in[1,n]}$ can obtain the corresponding training images dataset $\mathcal{X}_{i\in[1, n]}$ collected by the in-vehicle cameras. Therefore, the $\mathcal{X}_i$ is regarded as the training dataset to execute the following steps.
\item \textbf{Local Models Training:} Local models training is the second step to realizing the whole system and is also the process of generating a large of local predictive models. After the training dataset generation, the initialized global model parameter is sent to each VC and utilizes their images as a training dataset to train the local model. Similarly, each vehicle $V_i$ could train its local model's parameters set $LM_i$. Thereby, $n$ vehicles can compose a model vector $M = [LM_1, LM_2, \cdots, LM_n]$. To guarantee the security of generated local models' parameters, we adopt Paillier cryptography to prevent the submitted local models from being reversed by malicious VCs and attackers.
\item \textbf{Poisoning Defense Mechanism:} The poisoning defense mechanism is the core component for the HDDP designed to distinguish whether the submitted local models have been poisoned in the PDS. If the submitted models suffer from data poisoning attacks, the HDDP will clear the poisoned model to eliminate the impact on the integrated global model. To execute the poisoning defense mechanism, the PDS needs to construct the queue $\mathcal{Q}$ for receiving a large of submitted local models' parameter sets. According to the definition in the above rounds of local model training, the $\mathcal{Q}$ can be presented as  $\mathcal{Q} = [LM_{(1, k)}, LM_{(2, k)}, \cdots, LM_{(n, k)}]$, where the $k$ represents the count of the PDS receiving local model, which is the number of rounds during local model training. When $\mathcal{Q}$ is verified by the poisoning defense algorithm, the system will eliminate the poisoned local models' parameters set from $\mathcal{Q}$.
\item \textbf{Secure Model Announcement:} When the poisoned local models are eliminated from the queue $\mathcal{Q}$, the MAS will integrate the submitted benign local models $\mathcal{Q}'$ to generate the global model $GM$ based on the classical FedAvg aggregation algorithm, such as $\mathcal{Q}\stackrel{FedAvg}{\longrightarrow} GM$. And then, each benign VC receives the benign $GM$ from MAS and further uploads the new local model for the next round. 
\end{itemize}
\subsection{Secure Local Models Submission}
According to the system model mentioned in section IV, the VCT train the local models $LM$ via their images $\mathcal{X}$, such as.
\begin{equation}
\mathcal{X}\stackrel{\mathcal{F}}{\longrightarrow} LM
\end{equation}

Owing to exist of the malicious VCs or malicious attacker $\mathcal{A}$, the trained local models may be poisoned. Assume that there are $\tau\;(0<\tau<\frac{n}{2})$ poisoned local models, which can be indicated as $LM^{*}$ to form poisoned models queue $Q^{*}$. Therefore, the whole local queue is also expressed, such as $Q=[Q',Q^{*}]$. To ensure the safety of each local model transmission, we adopt Paillier cryptography to achieve the ciphertext for each local model's parameters. Specifically, the PDS randomly chooses two large $e-$bit primes $p_1,q_1$, which satisfies the least common multiple relationship $\lambda=lcm(p_1-1,q_1-1)$. And then, PDS calculates the product value of $N=p_1\cdot q_1$ and obtains the $\mu=(L(g^{\lambda} mod N^{2}))^{-1} mod\;N$ based on the function $L(x)=\frac{x-1}{N}$, where $g\in \mathbb{Z}^{*}_n$. Finally, the secret key ($\lambda,\mu$) and public key ($N,g$) are generated in PDS. Note that the PDS could send the ($\lambda,\mu$) and  ($N,g$) to each VC. To avoid the leakage of a private message from local models, it performs the following encryption steps before each vehicle $V_i$ sends their local model's parameters set $LM_i$ to the PDS:

I. Firstly, each $V_{i}$ utilizes the public key $(N, g)$ and randomly chooses a value $r (r\in \mathbb{Z}^{*}_{n})$ to compute the ciphertext $c_{LM_{i}}$ of local model $LM_{i}$, such as $c_{LM_{i}} = Enc(LM_{i}) = g^{LM_{i}}\cdot {r^{N}}\mod{N^{2}}$, where the element $m$ in local model $LM_i$ belongs to the domain is $\mathbb{Z}_n$.

II. When the $V_{i}$ obtains the ciphertext $c_{LM_{i}}$ for plaintext local model $LM_{i}$, the $c_{LM_{i}}$ is sent to the PDS.

III. Finally, when the MAS receives the ciphertext of benign local models, the MAS could obtain the ciphertext of global model $c_{GM}$ based on the additive homomorphic characteristics of Paillier. And then, the MAS sends the $c_{GM}$ to each $V_{i}$ which could use the secret key $(\lambda, g)$ to get the plaintext of $GM$, such as $GM = Dec(c_{GM})=L(c^{\lambda}_{LM_{i}})\cdot \mu\mod N$.
\subsection{Hierarchical Defense Data Poisoning Attacks}
In this section, we mainly state more detail about the HDDP system and further describe the execution of the HDDP. To better accomplish the data poisoning defense mechanism, we design two types of data poisoning defense mechanisms to defend against two different situations. One is that we design a local model test voting (LMTV) mechanism with the trusted test dataset $\mathcal{D}_{test}$ in PDS; another is that we also design a Kullback-Leibler divergence anomaly parameters detection (KLAD) mechanism without the $\mathcal{D}_{test}$.
\subsubsection{LMTV Defense Mechanism}
When a each $V_i$ submits its ciphertext $c_{LM_{i}}$, $n$ number $c_{LM_{i\in[1,n]}}$ could construct the ciphertext vector $c_{\mathcal{Q}}$, such like $c_{\mathcal{Q}} = [c_{LM_{1}}, c_{LM_{2}}, \cdots, c_{LM_{n-1}}, c_{LM_{n}}]$. Suppose that the PDS have a trusted test dataset $\mathcal{D}_{test}$, which is utilized to detect whether the submitted local models have suffered from a data poisoning attacks and further remove the poisoned VCs as follows.

I. Because the PDS is an authenticated and trusted party, PDS utilizes the secret key $sk = (\lambda, \mu)$ to decrypt the ciphertext vector $c_{\mathcal{Q}}$, such like $c_{\mathcal{Q}} = [c_{LM_{1}}, c_{LM_{2}}, \cdots, c_{LM_{n-1}}, c_{LM_{n}}]$ to obtain the plaintext models vector $\mathcal{Q} = [LM_{1}, LM_{2}, \cdots, LM_{n-1}, LM_{n}]$.

II. And then, the PDS evaluates each local model submitted to give the classification of the natural objects based on the trusted $\mathcal{D}_{test}$. And then, the classification accuracy rate $(CAR)$ for each object classification is computed as the following formula.
\begin{equation}
CAR = \frac{num(obj)-num(obj^{*})}{num(obj)}\cdot 100\%
\end{equation}
Where the $num(obj^*)$ and $num(obj)$ are the number of the error classification by the each local model, respectively. The value of $num(obj^*)$ and $num(obj)$ can be obtained by confusion matrices.

III. According to those mentioned above, the number of poisoned local models $\tau$ satisfies the relationship of $0< \tau< \frac{n}{2}$. Therefore, the PDS implements a defense mechanism against poisoned local models based on CAR's value. Specifically, when each local model $LM_{i}$ is detected, the PDS could calculate the value of $CAR$ for each classification based on $\mathcal{D}_{test}$. As for $CAR$ of each classification of local model $LM_{i}$ in $n$ local models, if the $50\% \le CAR$, the PDS regards the local model $LM_{i}$ detected at this time as a normal model and is retained. Conversely, if $0< CAR < 50\%$, the PDS regards the local model $LM_{i}$ detected as a poisoned model and is further eliminated.

IV. After detecting each local model, the PDS uses the public key $pk = (N, g)$ to encrypt the benign local models and obtain the ciphertext. Finally, the ciphertext $c'_{LM'_{i}}$ of many benign local models $LM'_{i}$ will be composing the benign vector $\mathcal{Q}'$, which is sent to MAS.

According to the above execution, we can achieve the data poisoning defense mechanism based on the trusted test dataset $\mathcal{D}_{test}$. The above process can be given as algorithm 1.
\begin{algorithm}[!htb]
	\caption{Local Model Test Voting (LMTV).}
	\KwIn{the task list $c_{\mathcal{Q}} = [c_{LM_{1}}, c_{LM_{2}}, \cdots, c_{LM_{n}}]$; the secret key $(\lambda, \mu)$; the public key $(N, g)$\;}
	\KwOut{benign local models vector $c_{\mathcal{Q}'} = [c_{LM'_{i}}]$\; }
	\For {$1\le epoch \le Epoches$}{
	\For {$1\le i \le n$}{
	PDS utilizes the $(\lambda, g)$ to decrypte the $c_{\mathcal{Q}}$, such as $LM_{i} = Dec(c_{LM_{i}}) = L(c^{\lambda}_{LM_{i}})\cdot\mu{\mod N}$, and constructs $\mathcal{Q} = [LM_{1}, LM_{2}, \cdots, LM_{n}]$\;
	}	
	\For {$1\le j \le n$}{
	PDS evaluates the $LM_{j}$ to get $CAR$ based on $\mathcal{D}_{test}$\;
	\If {$0 < CAR < 50\%$}{
	PDS deletes $LM_{j}$ from $\mathcal{Q}$;
	}
	\Else{
	PDS saves local model $LM'_{j}$ and obtains $c'_{LM_{j}}=Enc(LM'_{j})=g^{LM'_{j}}\cdot r^{N}\mod N^{2}$\;
	}
	}
	}
	PDS gets benign local models vector $\mathcal{Q'}$\;
	\textbf{Return} $\mathcal{Q'}$
\end{algorithm}
\subsubsection{KLAD Defense Mechanism} For the LMTV mechanism, the trusted test dataset $\mathcal{D}_{test}$ in PDS is a prerequisite for data poisoning defense. If there is no test dataset in PDS, we can not use the LMTV algorithm to defend against data poisoning attacks. To address the above issue, we give a new data poisoning mechanism named KLAD. The critical point of KLAD is to measure the correlation relationship between any two local models. Here, we provide the specific formula representation of the calculation formula for KL divergence, such as.
	\begin{equation}
		\kappa_{KL}(X||Y)=\mathcal{H}(X)-\mathcal{H}(Y)=\sum_{x\in X,y\in Y}p(x)\cdot log\frac{p(x)}{p(y)}
	\end{equation}
	where the $Y$ are the random variable sets, $\mathcal{H}(X,Y)$ is the joint probability distribution function of $X$ and $Y$. The $\mathcal{H}(X,Y)$ can be computed by the information entropy $\mathcal{H}(X)$ by the formula $\mathcal{H}(X)=-\sum_{x\in X}p(x)log(p(x))$ and the $\mathcal{H}(Y)$ can be get by $\mathcal{H}(Y)=-\sum_{y\in Y}p(y)log(p(Y))$.
	
To use the kullback-leibler divergence $\kappa_{KL}(\cdot)$ to measure the correlation relationship between two models, we calculate $\kappa$ for the model parameters of each layer of the network. More details about the KLAD can be given as follows.

I. Firstly, when the PDS receives the ciphertext vector $c_{\mathcal{Q}} = [c_{LM_{1}}, c_{LM_{2}}, \cdots, c_{LM_{n-1}}, c_{LM_{n}}]$,  the PDS will also decrypt the $c_{\mathcal{Q}}$ to obtain the plaintext $\mathcal{Q} = [LM_{1}, LM_{2}, \cdots, LM_{n-1}, LM_{n}]$ based on the secret key $(\lambda, \mu)$, such as $LM_{i} = Dec(LM_{i}) = L(c^{\lambda}_{LM_{i}})*\mu\mod N$.

II. And then, the plaintext for each local model's parameters can be divide into $L$ parts based on the number of layer of network model in the PDS. For instance, as for local models $LM_{i\in[1, n]}$, whose the network structure $\mathcal{W}$ has four layer $\mathcal{L}^{{LM}_{i}}_{\omega\in[1, 4]} = [\omega^{(LM_{i}, \mathcal{L}_{1})}_{1}, \omega^{(LM_{i}, \mathcal{L}_{1})}_{2}, \omega^{(LM_{i}, \mathcal{L}_{1})}_{3}, \omega^{(LM_{i}, \mathcal{L}_{1})}_{4}]$.

III. Next, according to the calculation form of $\kappa_{KL}(\cdot)$, we randomly select $\mathcal{L}^{{LM}_{i}}_{\omega\in[1, L]}$ to compute the the $\kappa_{KL}(\cdot)$ of each layer in $\mathcal{L}^{{LM}_{i}}_{\omega\in[1, L]}$ between with any other layer in $\mathcal{L}^{{LM}_{j}}_{\omega\in[1, L]}$. Specifically, we use the corresponding layer to respectively replace $X$ and $Y$ in equation $10$ and obtain the value of $\kappa_{KL}{(LM_{i}||LM_{j})}$. Note that the $i, j\in [1, n]$ and $i \neq j$. And then, the PDS computes the function $F$ to obtain a value reflecting the correlation of the entire model of $\mathcal{L}^{{LM}_{i}}_{\omega\in[1, L]}$ and $\mathcal{L}^{{LM}_{j}}_{\omega\in[1, L]}$. The function $F$ can be indicated as follows.
\begin{equation}
F(\mathcal{L}^{{LM}_{i}},\mathcal{L}^{{LM}_{j}}) = \frac{\sum^{L}_{i=1}{\kappa_{(KL,\;i)}}}{|L|}
\end{equation}

IV. According to step 3 above, the PDS could get $n$ number of $F$. To remove the poisoned local models' parameters, the $F's$ average value $\tilde{F}$ is regarded as an important judge indicator which is computed as $\tilde{F}=\frac{\sum^{n}_{i=1}{F}}{n}$. And then, the PDS obtain the different ratio $\rho$ by calculating the size of each $F_{i\in[1,n]}$ deviation from the $\tilde{F}$ as follows.
\begin{equation}
\rho = \frac{|F_i-\tilde{F}|}{|\tilde{F}|}
\end{equation}

V. Finally, according to the value of $\rho$ and threshold $\Theta$, the PDS achieve the defense mechanism for data poisoning. Specially, if the $\rho\le \Theta$, the PDS saves the benign local model $LM'_i$ and uses the public key $(N,g)$ to computes the cipher-text of $LM'_i$, such as $c_{LM'_i}=Enc(c_{LM'_i})=g^{LM'_i}\cdot r^{N} mod\;N^{2}$; otherwise, the $\Theta\le \rho$, the PDS regards the local clients as the poisoned local clients and further removes them.

Based on the above steps, we could realize the KLAD without the trusted dataset $\mathcal{D}_{test}$ to achieve the defense data poisoning attacks. Thereby, the above steps can be shown as algorithm 2, and the defense algorithm chose like the flow chart as Fig. 3.
\begin{figure}
	\centering
	\includegraphics{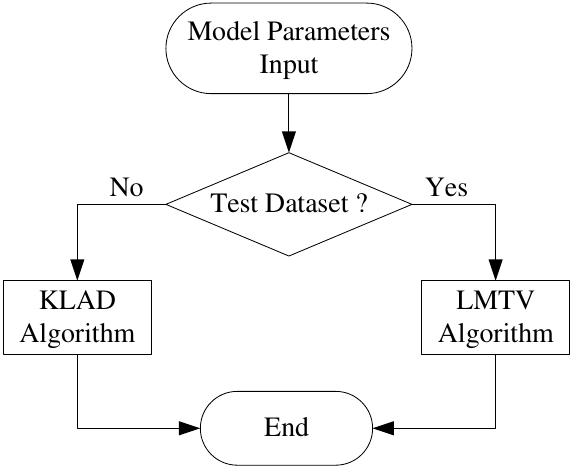}
	\caption{Flowchart of Defense Algorithm}
	\label{fig:my_label}
\end{figure}
\begin{algorithm}[!htb]
	\caption{Kullback-Leibler Divergence Anomaly Parameters Detection (KLAD).}
	\KwIn{the task list $c_{\mathcal{Q}} = [c_{LM_{1}}, c_{LM_{2}}, \cdots, c_{LM_{n}}]$; the secret key $(\lambda, g)$; the public key $(N, g)$\;}
	\KwOut{benign local models vector $c_{\mathcal{Q}'} = [c_{LM'_{i}}]$\;
	}
	\For {$1\le epoch \le Epoches$}{
	\For {$1\le i \le n$}{
	PDS gets $LM_{i} = Dec(c_{LM_i}) = L(c^{\lambda}_{LM_{i}})*\mu\mod N$, and further constructs $\mathcal{Q} = [LM_{1}, LM_{2}, \cdots, LM_{n-1}, LM_{n}]$\;
	}
	PDS randomly selects the local $LM_i$ from $\mathcal{Q}$;\\	
	\For{$1\le j \le n$}{
	\For {$1\le t \le L $}{
	PDS computes the KL divergence, $\kappa_{KL}=\mathcal{H}(\mathcal{L}^{LM_{i}}_{t})-\mathcal{H}(\mathcal{L}^{LM_{i}}_{t},\mathcal{L}^{LM_{j}}_{t})$\;
	}
	PDS gets $F(\mathcal{L}^{LM_{i}}, \mathcal{L}^{LM_{j}}) = \frac{\sum^{L}_{t = 1}\kappa_{(KL,t)}}{|L|}$\;
	}
	PDS calculates the average value $\tilde{F}=\frac{\sum^{n}_{i=1}F}{n}$, computes $\rho = \frac{|F_i-\tilde{F}|}{|\tilde{F}|}$, determine the threshold $\Theta$\;
	\If {$\Theta<\rho$}{
	PDS removes $LM_{i}$ from $\mathcal{Q}$;
	}
	\Else{
	PDS saves $LM'_{i}$, computes the ciphertext $c_{LM'_{i}} = Enc(LM'_{i})=g^{LM'_{i}}\cdot r^{N}\mod N^{2}$\;
	}
	}
	PDS gets benign local models vector $\mathcal{Q'}$\;
	\textbf{Return} $\mathcal{Q'}$
\end{algorithm}
\subsection{Secure Global Aggregation}
After execution based on LMTV and KLAD, the PDS removes the poisoned local models and saves the benign local models. Suppose that the MAS receives the ciphertexts for $m$ benign local models, such as $\mathcal{Q}' = [c_{LM'_{1}}, c_{LM'_{2}}, \cdots, c_{LM'_{m}}]$, where the $m$ is the number of benign local models submitted. The MAS can achieve the aggregation for a secure global model in the following steps.

I. As for each $c_{LM'_{i}}$, the MAS computes the value of $m$ root, such as $(Enc (LM'_{i}))^{\frac{1}{m}}={(c_{LM'_{i}}})^{\frac{1}{m}}$.

II. When the MAS gets ${(c_{LM'_{i}}})^{\frac{1}{m}}$. and further computes the product of $m$ models ciphertext to obtain the ciphertext of global model $c_{GM} = \prod_{i=1}^{m}{(Enc (LM'_{i}))^{\frac{1}{m}}}$.

III. Finally, each stored VC receives the $c_{GM}$ from the MAS, uses the secret key $(\lambda, \mu)$ to obtain the plaintext of global model $GM$, and executes a straining learning algorithm to update its local model's parameters.

In order to illustrate the ciphertext of global model $c_{GM} = \prod_{i=1}^{m}{(Enc (LM'_{i}))^{\frac{1}{m}}}$, and ensure that the plaintext of $c_{GM}$ satisfy the result calculated by FedAvg algorithm, we give the proof process as follows.
\begin{large}
\[\begin{array}{l}
c_{GM} = \prod_{i=1}^{m}{(c_{LM'_{i}}})^{\frac{1}{m}}\\
\;\;\;\;\;\;\;=(c_{LM'_{1}})^{\frac{1}{m}}\times (c_{LM'_{2}})^{\frac{1}{m}}\times \cdots \times (c_{LM'_{m}})^{\frac{1}{m}}\\
\;\;\;\;\;\;\;=Enc (\frac{LM'_{1}}{m})\times Enc (\frac{LM'_{2}}{m})\times \cdots \times Enc (\frac{LM'_{m}}{m})\\
\;\;\;\;\;\;\;=Enc(\frac{LM'_{1} + LM'_{2} + \cdots + LM'_{m}}{m})\\
\;\;\;\;\;\;\;=c_{GM}
\end{array}\]
\end{large}

\section{Security \& Performance Analysis}
To illustrate the security of the proposed defense algorithms based on the defined attack model, we give the security and performance analysis of the HDDP in this section.
\subsection{Security Analysis}
According to the definition of the attack model, two different directions of data poisoning attacks exist to influence the performance of the global model. A type of data poisoning attack comes from the malicious VCs, and another type of data poisoning attack is from the malicious attacker $\mathcal{A}$, which has the same hazard of misclassifying data to lead to the classification error and accuracy reduction of the global model. For instance, the poisoned local model could classify the source target $0$ into poisoning target $2$. To illustrate the global model's effectiveness and the defense against the attacks of malicious VCs and attackers $\mathcal{A}$, the proofs are given to state the correctness of the HDDP system.
\begin{theorem}
\emph{The HDDP system effectively guarantees the security of local models, ensuring that the local models' parameters are not leaked during the transmission process. Even if the local models are stolen during transmission to PDS, the attackers $\mathcal{A}$ cannot obtain any private message.}
\end{theorem}
\begin{proof}
As the content mentioned above, the system can generate $\mathcal{M}$ local models based on $\mathcal{M}$ VCs. To protect the transmission of each local model's parameters, the third trusted party PDS generates the key $(\lambda, \mu)$ and $(N, g)$ according to the Pallier cryptography theory are sent to each VC. Specifically, when the local models are generated, the VCs utilize the $pk = (N, g)$ to compute the ciphertext $c_{LM_i}$ of each local model, such as $c_{LM_{i\in[1,\mathcal{M}]}} = Enc(LM_i) = g^{LM_i}\cdot r^{N}\mod N^{2}$. From the security of Paillier cryptography \cite{ref57}, the transmission process could guarantee the safety of local models' parameters. Even if the malicious attacker $\mathcal{A}$ obtains the ciphertext $c_{LM_i}$, he can not infer the whole plaintext of local models. More importantly, when the ciphertext $c_{LM_i}$ is sent to the PDS, all data processing is performed offline in PDS. Therefore, the HDDP system can ensure VCs' private information is secure.
\end{proof}
\begin{theorem}
\emph{The HDDP effectively defense the data poisoning attacks of malicious VCs to remove the poisoned local models. Based on the HDDP defense mechanisms, the final global model will not be affected in FL, which effectively defense the attack of malicious VCs.}
\end{theorem}
\begin{proof}
In this paper, we design two types of defense algorithms called LMTV and KLAD to defend against FL data poisoning attacks.
\subsubsection{Proof of LMTV}
As for the LMTV algorithm, the system mainly uses the trusted test dataset $\mathcal{D}_{test}$ to verify each local model for the threat of data poisoning attacks. Specifically, when the PDS receives the ciphertext of the $n$ local models, it goes offline to calculate the plaintext of each local model via $LM_{i\in[1,\mathcal{M}]} = Dec(c_{LM_i}) = L(c^{\lambda}_{LM_i})\cdot\mu\mod N$ to construct the model vector $\mathcal{Q}=[LM_1, LM_2, \cdots, LM_n]$. Because the PDS is the trusted third party and the above operations are performed offline, the PDS can not reveal the leakage of any private messages about each local model's parameters and protect the malicious attacker $\mathcal{A}$ from obtaining the plaintext information through the internet. To realize the defense algorithm for data poisoning attacks, the LMTV needs to achieve the classification confusion matrix to estimate every classification for every categorical label respectively based on the $\mathcal{D}_{test}$ in each epoch. For $n$ local models in every round, the system computes $CAR$ to obtain the each local model voting classification with different rates by calculating $CAR = {num(obj)-num(obj^*)}/{num(obj)}\cdot 100\%$. When $CAR>50\%$ means inconsistent classification voting with $n$ local models, the system will consider the few inconsistent ones as poisoned local models to execute the elimination operation. Since the number of prescribed poisoned local models, $\tau$ is between $0 < \tau < \frac{n}{2}$, the system can use the LMTV algorithm both to avoid the effects of local dataset differences and to eliminate poisoned local models that are inconsistent with the majority of model classifications.
\subsubsection{Proof of KLAD}
Different from the LMTV algorithm, KLAD can realize the poisoning defense without the test dataset $\mathcal{D}_{test}$. Here, the system mainly utilizes the KL divergence to calculate the correlation relationship between different local models. According to the attribution, the correlation between the poisoned local models and the friendly model shows a different relationship; the deeper the model's poisoning, the weaker the correlation between it and the benign local models. Specifically, when the PDS receives the ciphertext of local models $c_{Q}$, the PDS decrypts the ciphertext vector $c_{Q}$ to obtain the plaintext vector of local models $\mathcal{Q}$ via the private key $sk = (\lambda, \mu)$. Owing to the above operations being performed offline, the PDS is protected from the malicious attacker $\mathcal{A}$ to steal the private message sourcing from each local model's parameters. To accomplish the defense algorithm, the PDS randomly selects the local model $LM_i$ to divide the plaintext of the local model into $L$ parts based on the number of layers in-network $\mathcal{W}$. And then, the $\kappa_{KL}(\cdot)$ of each layer is also computed to further obtain the metric function $\kappa_{KL}(\cdot)$ between local model $LM_i$ and $LM_j$ in each epoch, such as $\kappa_{KL}(LM_i||LM_j) = \mathcal{H}(LM_i)-\mathcal{H}(LM_i,LM_j)$. In addition, the function $F$ is also computed via the formula $F = \frac{\sum_{i=1}^{L}\kappa_{(KL,i)}}{|L|}$, and the judge indicator $\tilde{F}$ is also obtained by computing $\tilde{F}=\frac{\sum{F_i}}{n}$. According to different ratios $\rho=\frac{|F_i-\tilde{F}|}{|\tilde{F}|}$ and the threshold $\Theta$, the KLAD can treat models with weakly correlated rows as poisoned models and further eliminate them. Specifically, if the $\Theta \le \rho$, PDS removes the current detected local model; otherwise, it saves the current model and encrypts it to get ciphertext. Therefore, the KLAD algorithm is also effective in defense against data poisoning attacks and guarantees the security of local models.
\end{proof}
\subsection{Performance Analysis}
To illustrate the effectiveness of the HHDP system, we utilize the standard Fashion-MINIST and GTSRB datasets for performance demonstration, concluding the influence of data poisoning attacks and two defense algorithms. The experiment environment with an Intel(R) Core(TM), i7-2600 CPU @3.40GHz, and 16.00GB of RAM, RTX3090 of GPU, is used to serve our central computational server. The parameters in HHDP are set up as the VCT number $\mathcal{M}=50$. To show the performance of the HDDP system, we analyze the computational cost and give experimental simulations of the LMTV algorithm and KLAD algorithm based on the label-flipping attacks in data poisoning attacks.
\subsubsection{Computation Cost}
For the computation cost of the HDDP system, we firstly analyze the computation cost for local models training, encryption and decryption of local models, and global aggregation of benign models. Assume that we define $\epsilon$ as the time for training an image based on the machine learning network and the $R$ represents the number of rounds per local model training. Similarly, we define $\nu$ as a local model's encryption and decryption processes. Therefore, the time for each local model submitted can be indicated as follows.
\begin{equation}
T = \epsilon\times \frac{\mathcal{D}_i}{|\mathcal{D}_i|}\times R + \nu
\end{equation} 
Besides, the computational cost of the whole HHDP mainly focuses on LMTV and KLAD defense algorithms. As for the LMTV defense algorithm, the computational cost is directly related to the number of models with a linear relationship. Based on this algorithm, assume that we give the $\hat{t}$ to express the computational cost for the system to achieve the defense of a submitted local based on LMTV. According to those mentioned above, the computational cost required for $\mathcal{M}$ models can be calculated, such as $\hat{t}\times\mathcal{M}$. Similarly, we define the $\bar{t}$ to indicate the computational cost for the system to realize the defense of a submitted local model based on the KLAD algorithm. For the $\mathcal{M}$ submitted local models, the total computational cost can be computed by $\bar{t}\times\mathcal{M}$. Additionally, we also definite $\eta$ as the time required for unit model aggregation. Therefore, the entire computational cost can get as follows. If the PDS has the trusted test dataset $\mathcal{D}_{test}$, the system will adopt the LMTV algorithm to achieve the defense against data poisoning, which requires a total computational cost $\hat{T}_{total}$ can be obtained by the equation $\hat{T}_{total}=(\hat{t}+T)\times\mathcal{M}+\eta\times \mathcal{M'}$. Otherwise, the PDS does not have the $\mathcal{D}_{test}$, the total computational cost $\bar{T}_{total}$ can be get by the equation $\bar{T}_{total}=(\bar{t}+T)\times\mathcal{M} + \eta\times \mathcal{M'}$, where the $\mathcal{M'}$ means the number of benign local models after defending data poisoning attacks.
\subsubsection{Impact of Data Poisoning Attacks}
To show the damage of data poisoning attacks, we make $0$ in turn with the classification of $1\sim 9$ via label-flipping attacks based on the convolutional neural network (CNN) at different poisoning rates in FL, such as Fig. 3. From Fig. 3(a) to Fig. 3(e), we show the results of no label-flipping attacks, $10\%$ label-flipping attacks, $20\%$ label-flipping attacks, $30\%$ label-flipping attacks, and $40\%$ label-flipping attacks by using the Fashion-MINIST dataset. According to Fig. 3(a) $\sim$ Fig. 3(e), we know that the value of $Pr$ gradually increases as the label-flipping attacks deepen, which illustrates that the label-flipping attacks effectively cause purposeful misclassification. When the label-flipping attack reaches $40\%$, any label-flipping attacks between $0$ and $1 \sim 9$ can achieve at least $87.336\%$ of error classification. To be more intuitive about the changes in label-flipping attacks of different rates, we can also summarize that $Pr$ has changed to varying degrees for any label-flipping attacks based on different rates of label-flipping attacks, which also matches the degree of similarity between the characteristics of different classifications. To illustrate the success rate of data poisoning, we subtract the $Pr$ after different label-flipping attacks ratios from the original $Pr$ based on the label-flipping attack to generate Fig. 3(f) and Fig. 3(g). The Fig. 3(f) means the change between Fig. 3(a) and Fig. 3(b), which shows that misclassification did not increase by less than $1\%$  when the label-flipping attacks are $10\%$ and $20\%$ except the $0$ and $6$ categories of attacks also reaches about $6\%$. From Fig. 3(g), we know that when the label-flipping attacks rate reaches $30\%$, the difference of $Pr$ has been marked further. It must be noted that when the label-flipping attacks rate reaches $40\%$, the success rate of label-flipping attacks between $0$ and $1 \sim 9$ has increased significantly, and the lowest values have reached $85.657\%$. Besides, we also utilize the GTSRB dataset to show the impact of the label-flipping attacks in the AV area, such as Fig. 3(h). From Fig. 3(h), the AV may incorrectly identify the speed limit signs and people as the roads. To remove the impact of the data poisoning attacks, we give the simulation for LMTV and KLAD in FL as follows.
\begin{figure*}[ht]
\centering
\subfloat[No Clients with Data Poisoning Attacks.]{
	\begin{minipage}{4cm}
		\includegraphics[scale = 0.32]{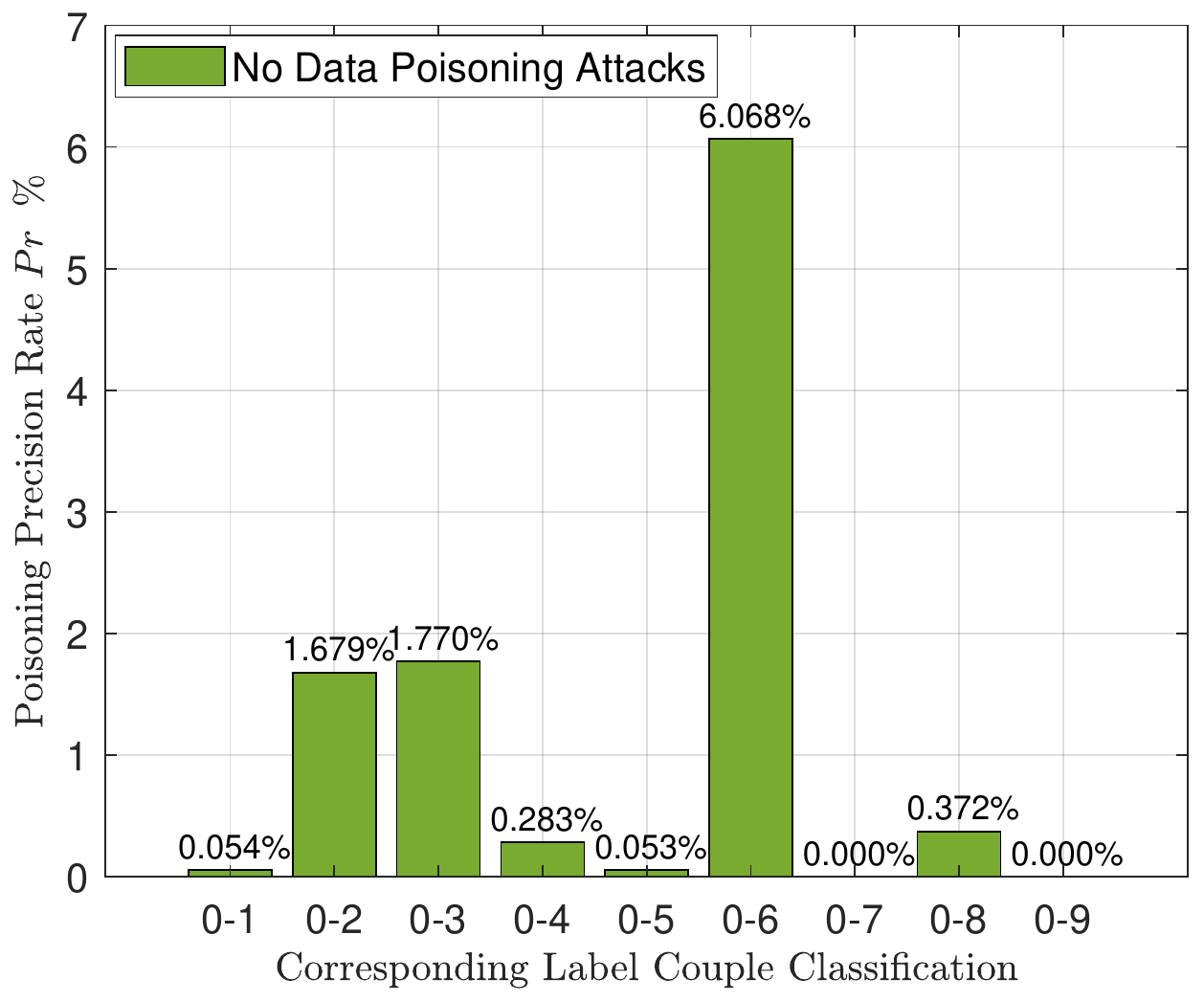}
	\end{minipage}}
\subfloat[10\% Clients with Data Poisoning Attacks.]{
	\begin{minipage}{4cm}
		\includegraphics[scale = 0.32]{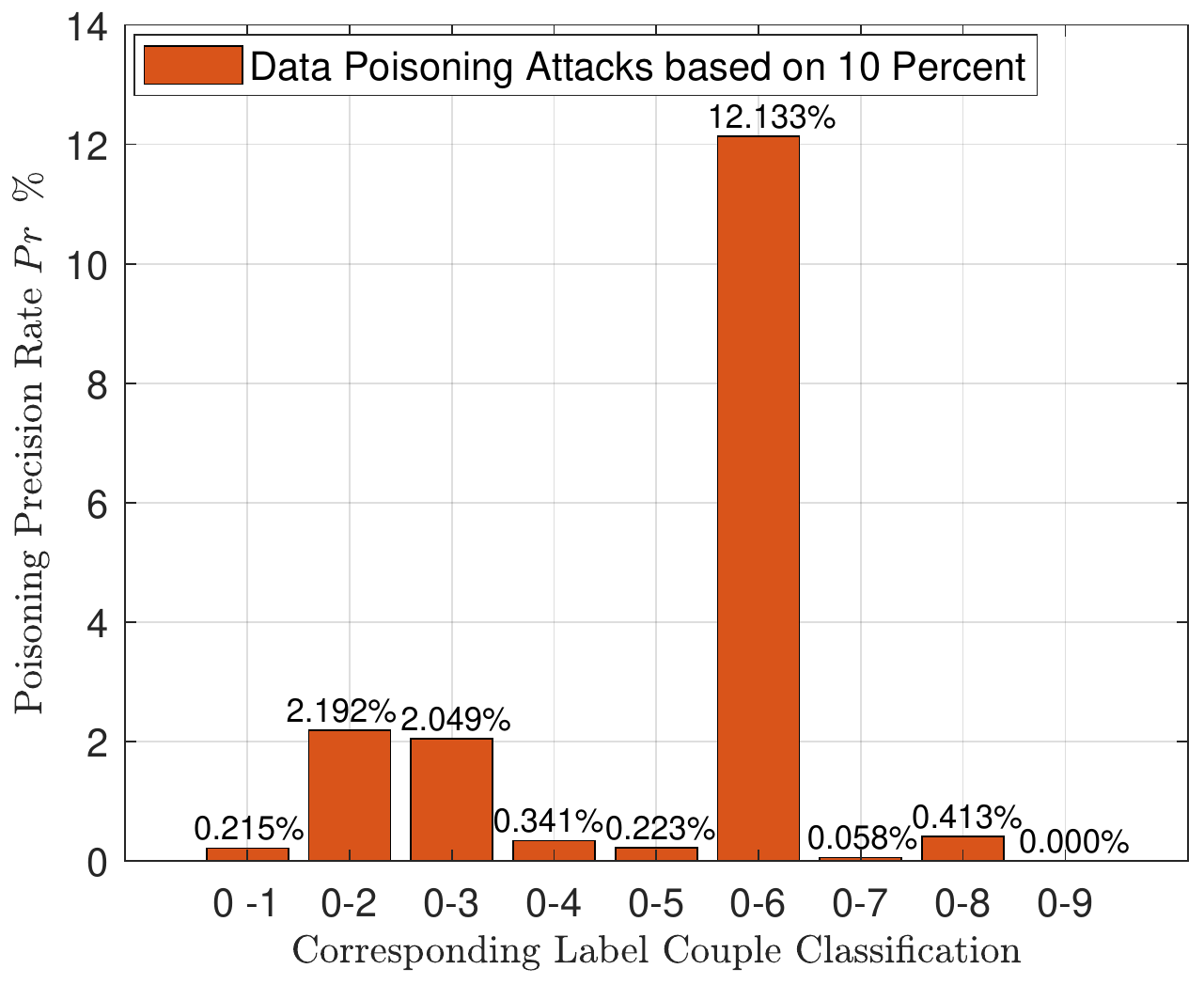}
	\end{minipage}}
\subfloat[20\% Clients with Data Poisoning Attacks.]{
	\begin{minipage}{4cm}
		\includegraphics[scale = 0.32]{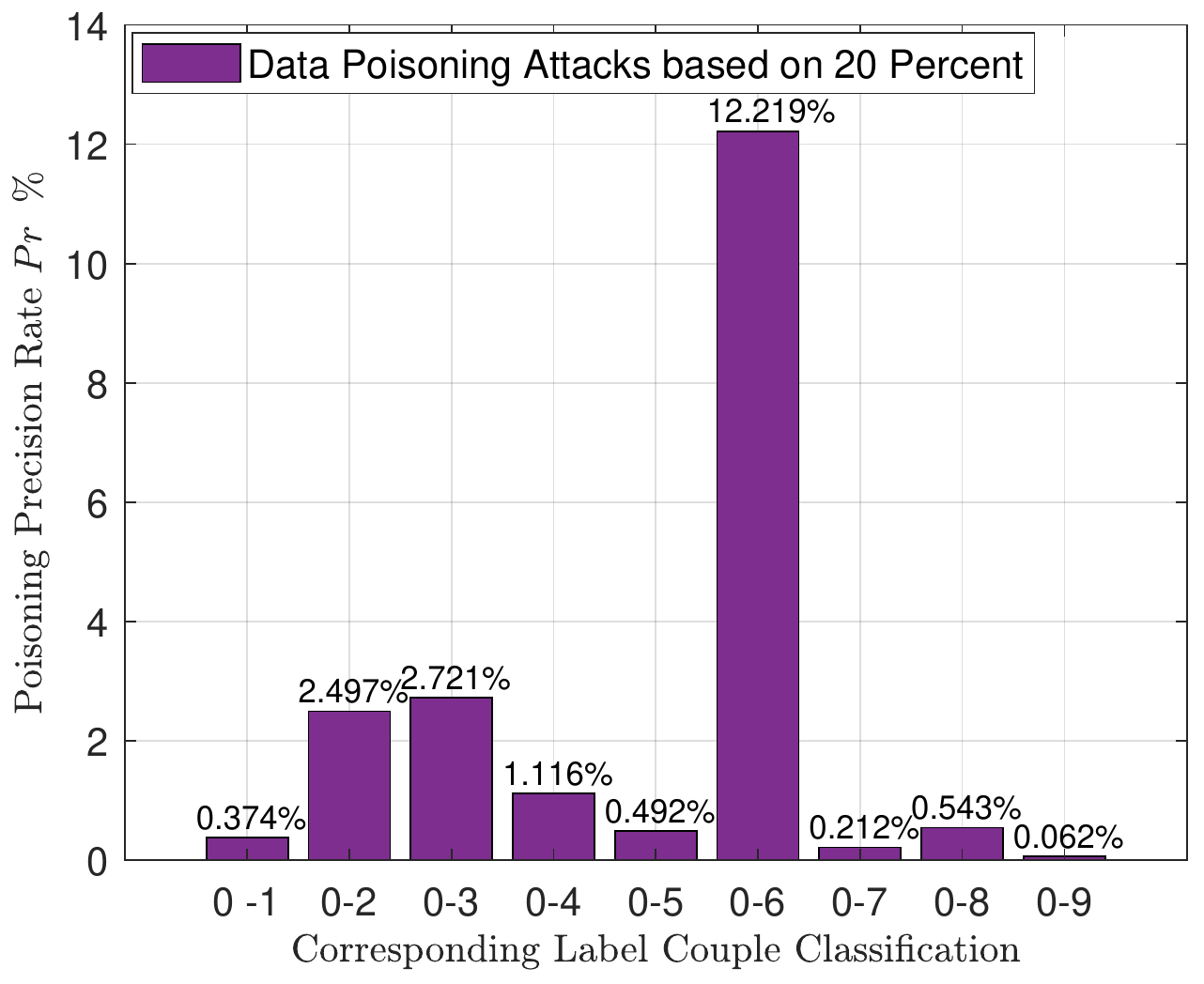}
	\end{minipage}}
\subfloat[30\% Clients with Data Poisoning Attacks.]{
	\begin{minipage}{4cm}
		\includegraphics[scale = 0.32]{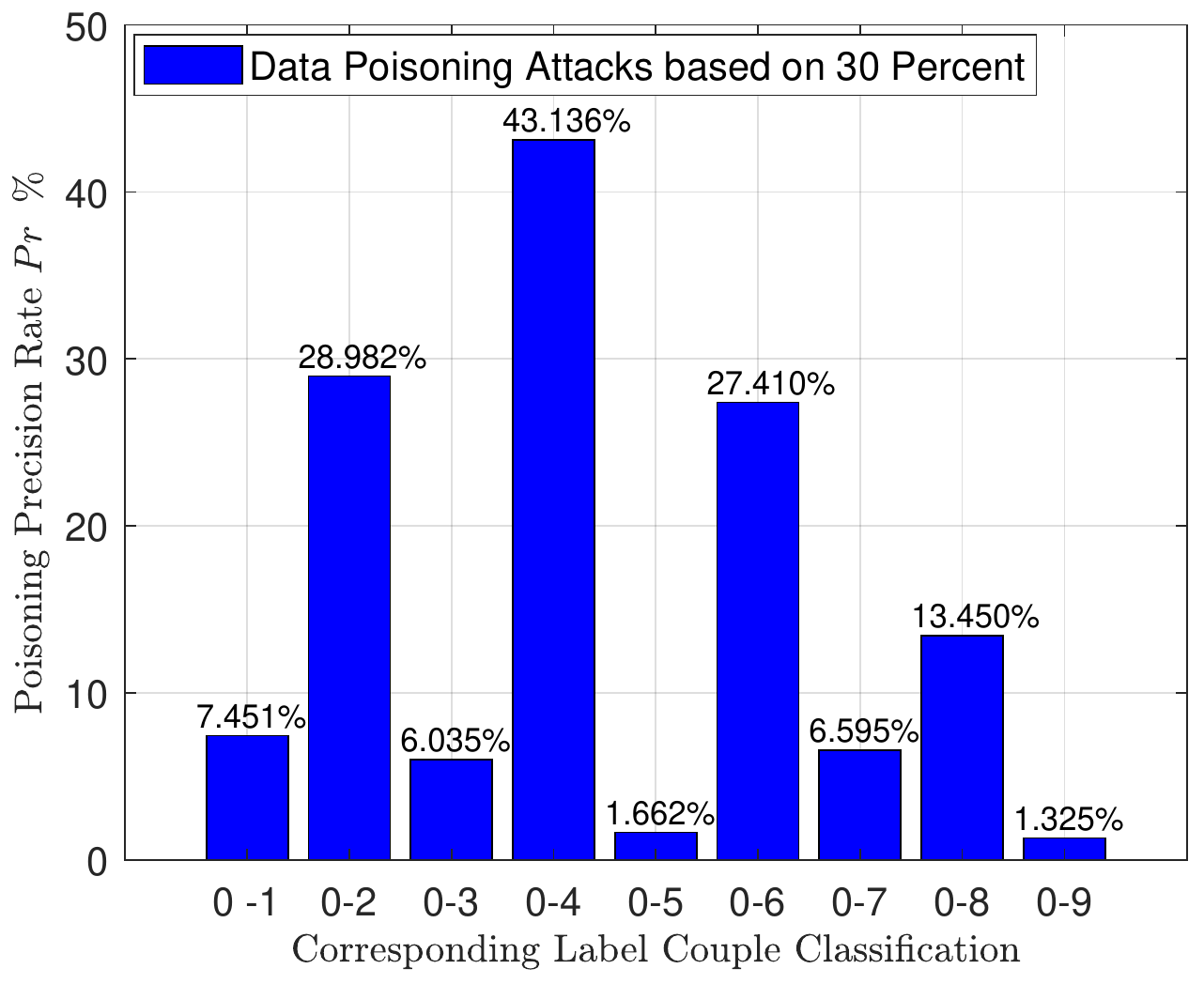}
	\end{minipage}}\\
\subfloat[40\% Clients with Data Poisoning Attacks.]{
	\begin{minipage}{4cm}
		\includegraphics[scale = 0.32]{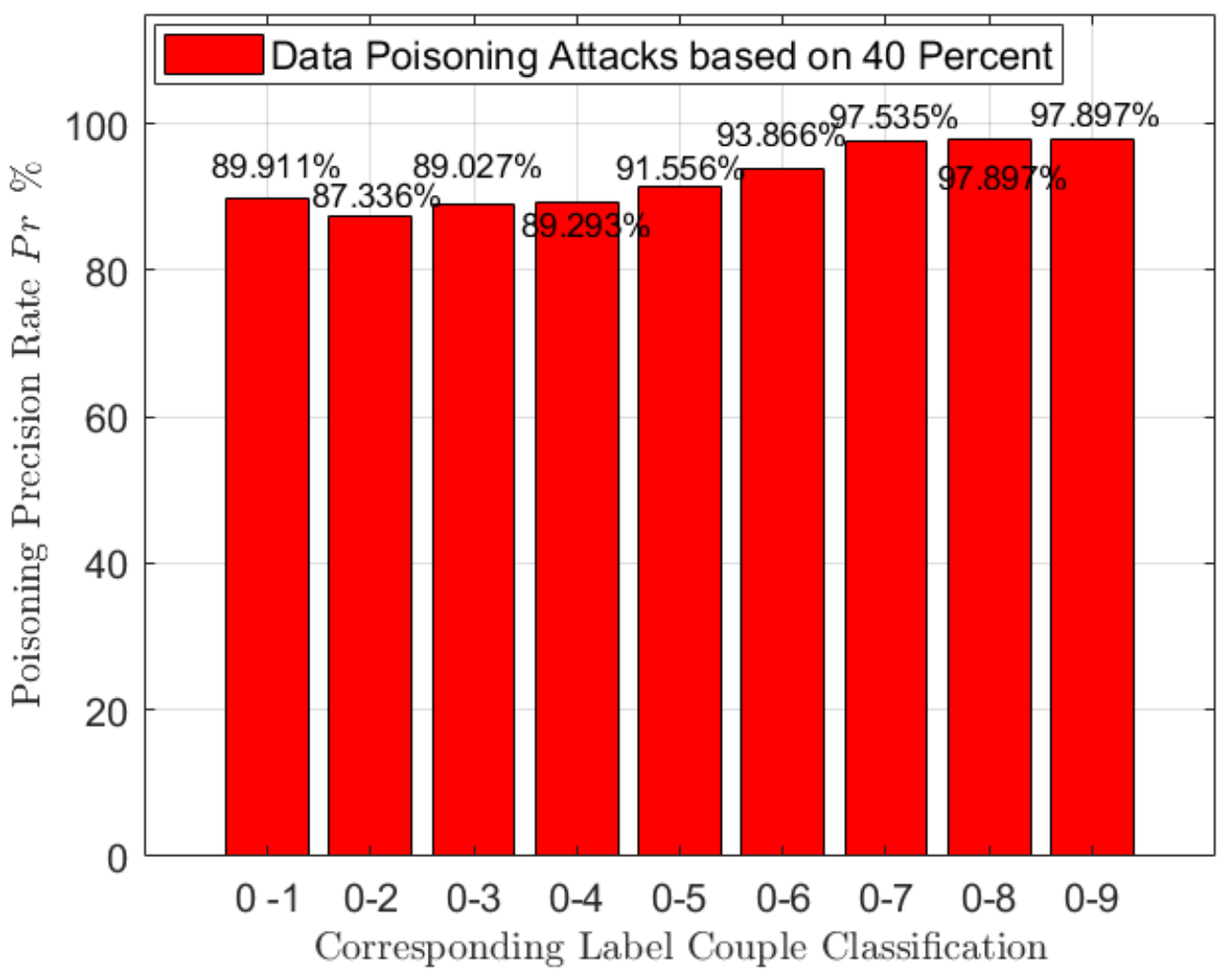}
	\end{minipage}}
\subfloat[Difference Value between (a) and (b), (c).]{
	\begin{minipage}{4cm}
		\includegraphics[scale = 0.32]{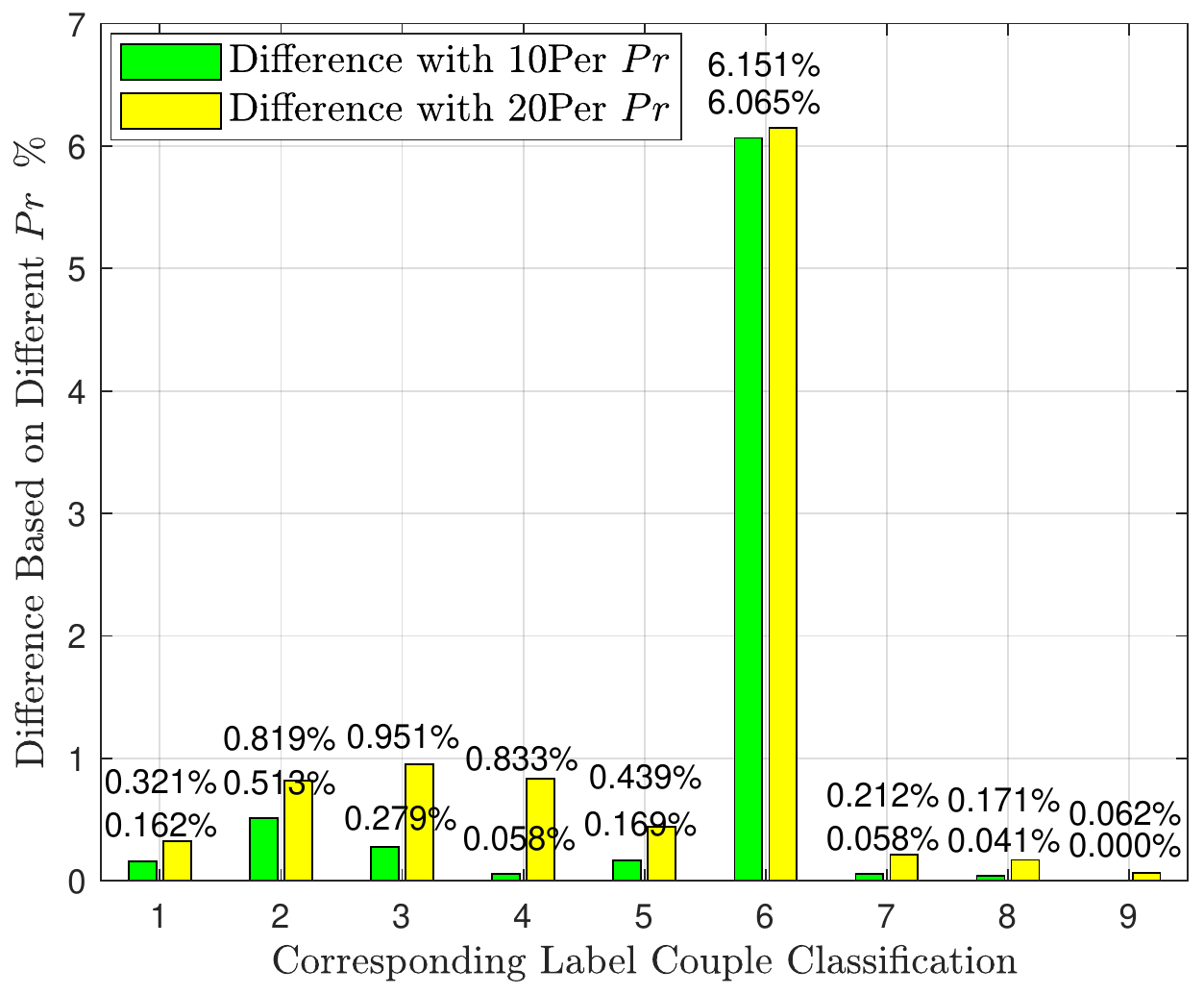}
	\end{minipage}}
\subfloat[Difference Value between (a) and (d), (e).]{
	\begin{minipage}{4cm}
		\includegraphics[scale = 0.32]{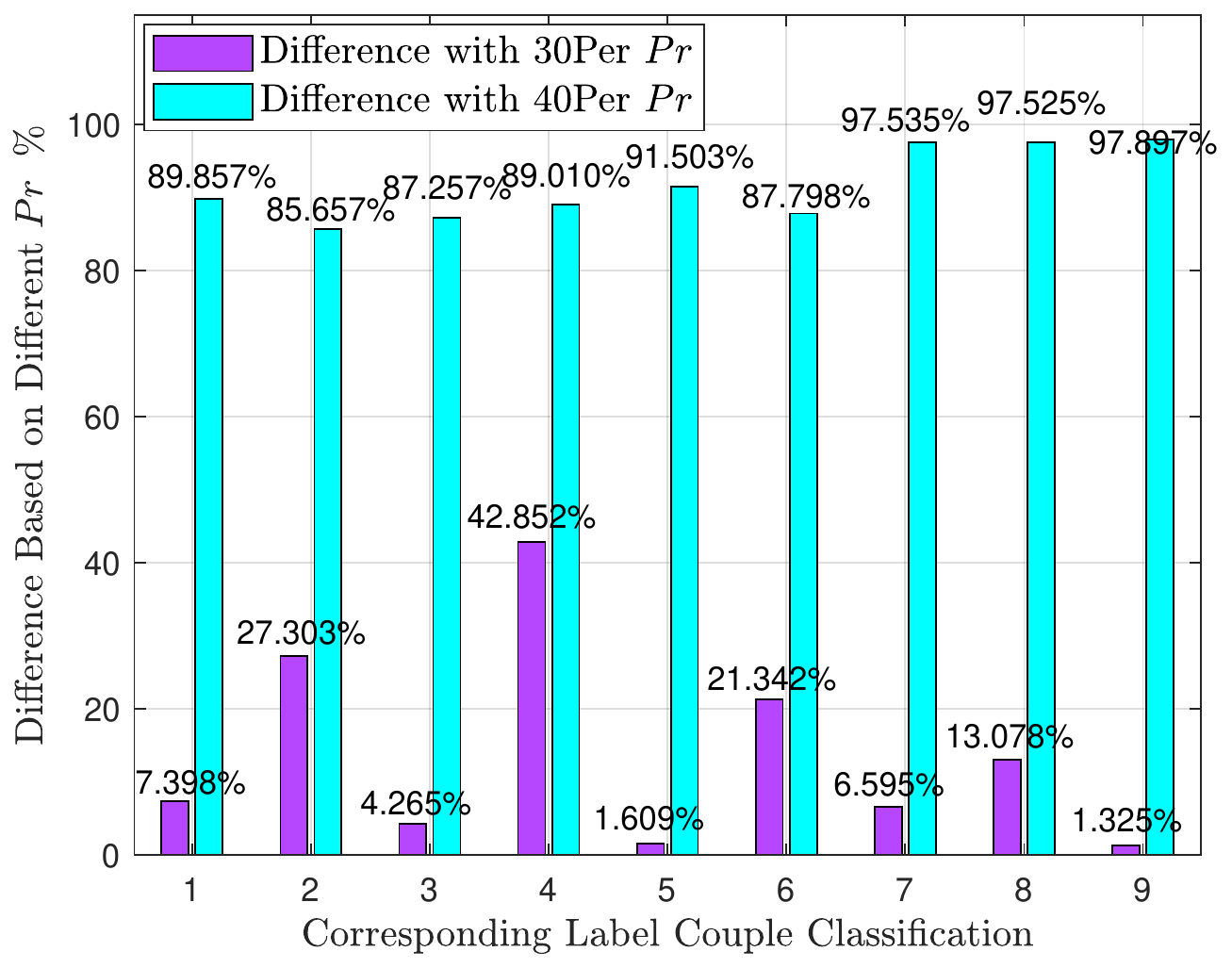}
	\end{minipage}}
\subfloat[Recognition Error among Classifications using GTSRB Dataset.]{
	\begin{minipage}{4cm}
		\includegraphics[width=4.1cm,height=3.33cm]{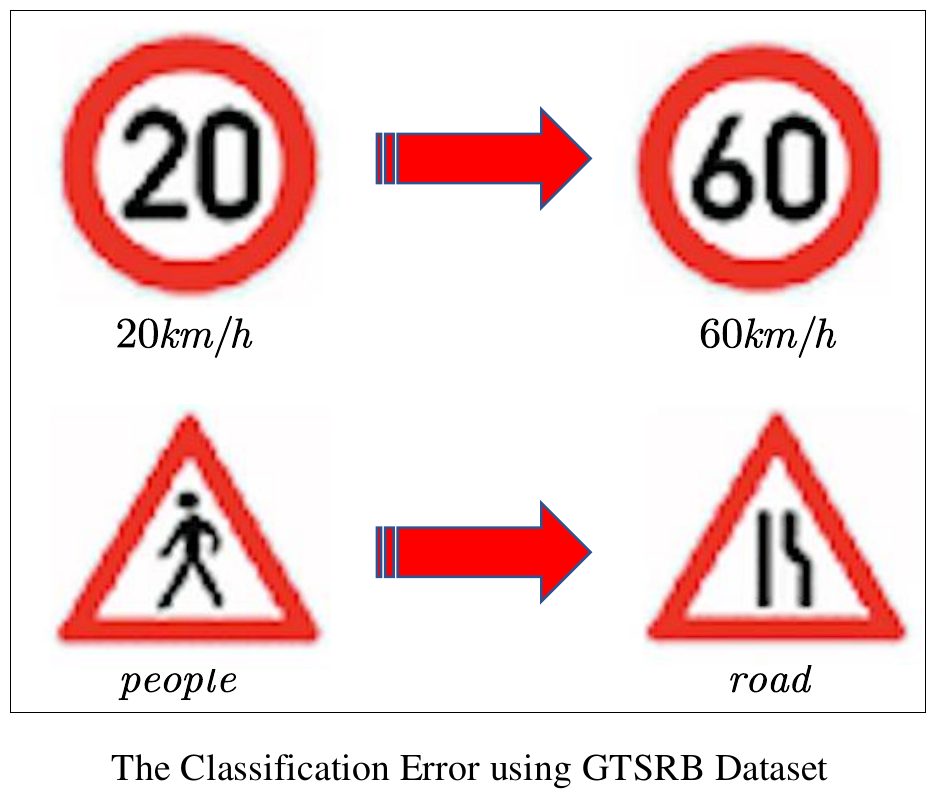}
	\end{minipage}}
\caption{Impact of Poisoning Attacks based on Different Rate by Flipping the Classification Labels between 0 and $1\sim 9$.}
\end{figure*}
\subsubsection{Function of LMTV}
To illustrate the function of LMTV, we utilize the LMTV algorithm to defend against the above label-flipping attacks based on Fashion-MINIST. The poisoning defense algorithm's key is divided into four different data poisoning attack rates $\gamma$, such as $\gamma\in\{10\%, 20\%, 30\%$, $40\%\}$, which can be computed as the number of poisoned local clients. According to the literature \cite{ref46}, we define the defense success rate $\mathcal{J}$  to measure the success of the defense algorithm based on the number of data poisoned local clients detected as $\mathcal{B}$, the total local clients $\mathcal{M}$. Besides, the number of data poisoning clients is $\mathcal{S}_i$ based on $\gamma$ by $\mathcal{S}_i=\gamma\times\mathcal{M}$. Therefore, the $\mathcal{J}$ can be expressed by the following.
\begin{equation}
\mathcal{J}=\frac{\mathcal{B}}{\mathcal{S}_i}=\frac{\mathcal{B}}{\gamma\times\mathcal{M}}\times100\%
\end{equation}

To improve the performance of the LMTV algorithm, we use the $\beta\;(\beta\in[0, 1])$ to control the detection threshold and adjust the sensitivity of data poisoning attacks detection. Note that the smaller value of $\beta$, the higher accuracy of the LMTV algorithm for detecting the poisoned local models. Of course, the smaller value of $\beta$ is not completely helpful for the defense data poisoned local models, and it also has a significant impact on benign local models. Therefore, we should choose an appropriate value of $\beta$ to achieve the defense against the poisoned model and not affect the benign local models. To determine the appropriate value of $\beta$, we give the defense ability for $40\%$ poisoned local models under different $\beta$ as shown in the following Fig. 4.
\begin{figure}[!htb]
	\centering 
	\includegraphics[width = 7cm]{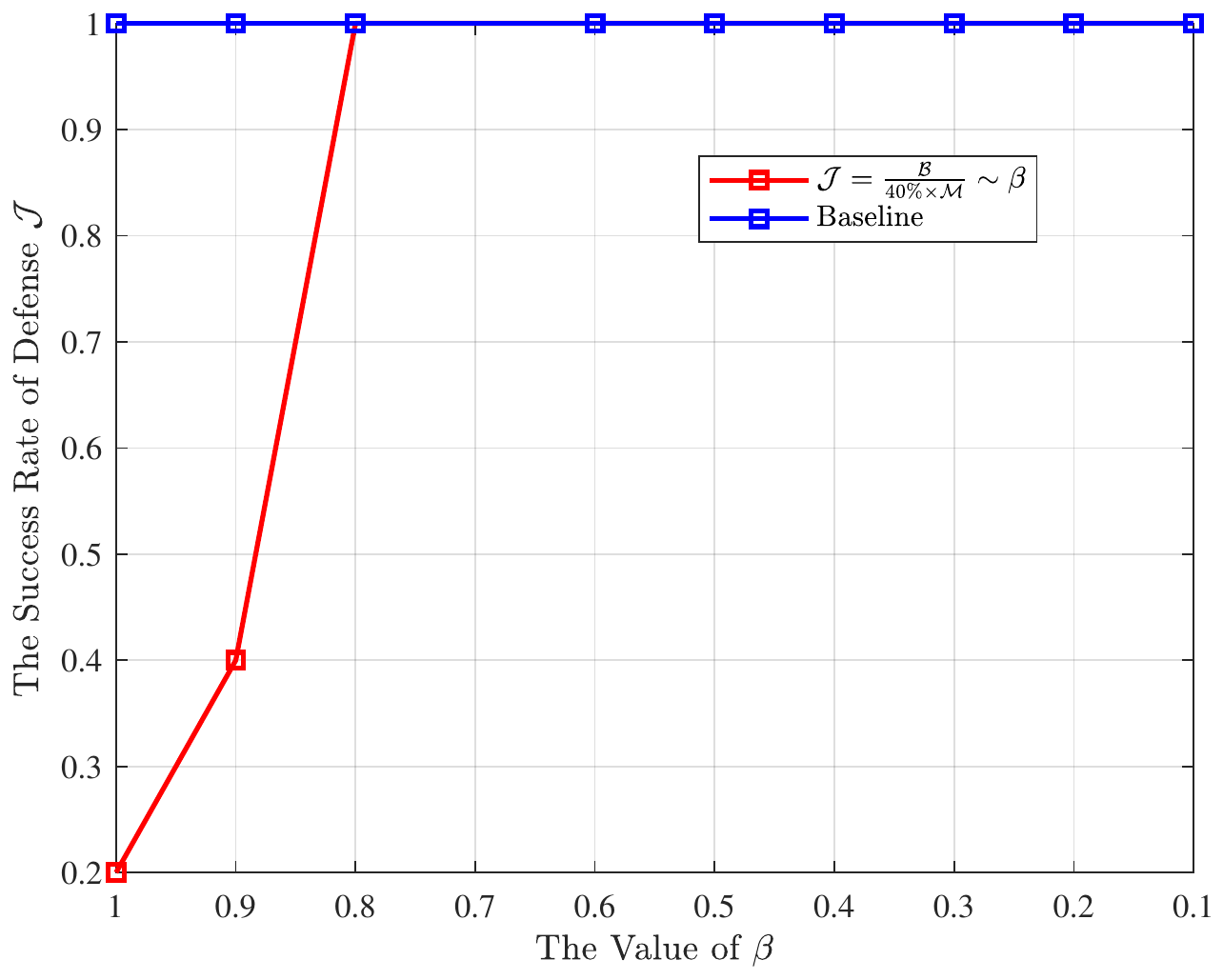}
	\caption{The Success Rate $\mathcal{J}$ with Different $\beta$.}
	\label{fig:figure1label}
\end{figure}
From Fig. 4, we can know the success rate of defense $\mathcal{J}$ has improved with the value of $\beta$ decrease. Additionally, when the $\beta = 0.8$, the $\mathcal{J}$ is constant to intersect with the blue baseline, which illustrates that the defense algorithm is optimal. The next comparative experiment is achieved based on $\beta = 0.8$. More importantly, the LMTV algorithm could perform under the four different data poisoning rates $\gamma$, such as Table I.
\begin{table}[!htb]
    \setlength{\tabcolsep}{1.6mm}{
	\begin{tabular}{lcccc}
		\hline
		Poisoning Rate $\gamma$ & $\gamma = 10\%$ & $\gamma = 20\%$ & $\gamma = 30\%$ & $\gamma = 40\%$ \\ \hline
		Poisoning VCs       & 5               & 10              & 15              & 20              \\ 
		Remove Malicious VCs     & 5               & 10              & 15              & 20              \\
		Remove Benign VCs		&0					&0					&0					&0				\\
		Defense Success Rate         & 100\%           & 100\%           & 100\%           & 100\%           \\ \hline
	\end{tabular}}
\caption{Performance of LMTV under $\gamma$ Data Poisoning Attacks}
\label{tab1}
\end{table}

To exhibit the performance of LMTV algorithm, we compare it with , Deep $k-$NN \cite{ref44}, One-Class SVM Defense \cite{ref58}, L2-Norm Outlier Defense \cite{ref59} based on the standard test dataset $\mathcal{D}_{test}$ from the Fashion-MINIST. The following table I compares the effectiveness of different defensive schemes in data poisoning attacks. Table II shows that four different defense strategies effectively detect data poisoning attacks. However, our LMTV algorithm is better than the other three methods.
\begin{table}[!htb]
    \setlength{\tabcolsep}{1.6mm}{
	\begin{tabular}{@{}lcc@{}}
		\toprule
		Defense Strategy    & Remove Poisoning Models & Defense Success Rate \\ \midrule
		\textbf{LMTV}       & \textbf{20/20}          & \textbf{100\%}       \\
		Deep $k-$NN ($k=5$) & 19/20                   & 95\%                 \\
		L2-Norm Outlierts   & 11/20                   & 55\%                 \\
		One-class SVM       & 07/20                   & 35\%                 \\ \bottomrule
	\end{tabular}}
\caption{Different Defense Strategies Compare under 40\% Data Poisoning Attacks}
\label{tab2}
\end{table}
\subsubsection{Function of KLAD} The LMTV algorithm cannot realize the defense strategy if the PDS does not obtain the trust test dataset. Here, we should rely on the KLAD algorithm to detect poisoned local models. To verify the performance of KLAD, we also adopt the Fashion-MINIST dataset based on the data poisoning rates $\gamma$ to simulate the experiment. Firstly, we give the threshold $\Theta$ to adjust the sensitivity of KLAD based on the $\mathcal{J}$, which means that a good $\Theta$ could improve the defensive effect of the KLAD defense algorithm. Fig. 5 shows the different performances based on different $\Theta$ under the data poisoning rate $\gamma = 40\%$. From Fig. 5, we can see that when the $\Theta < 1$, the value of $\mathcal{J}$ is over $1$, which means the benign local models are also regarded as the poisoned local models. To avoid excessive non-poisoning local models being deleted, we set the threshold $\Theta$ to $1$ to execute the following experiments.
\begin{figure}[!htb]
	\centering 
	\includegraphics[width = 7cm]{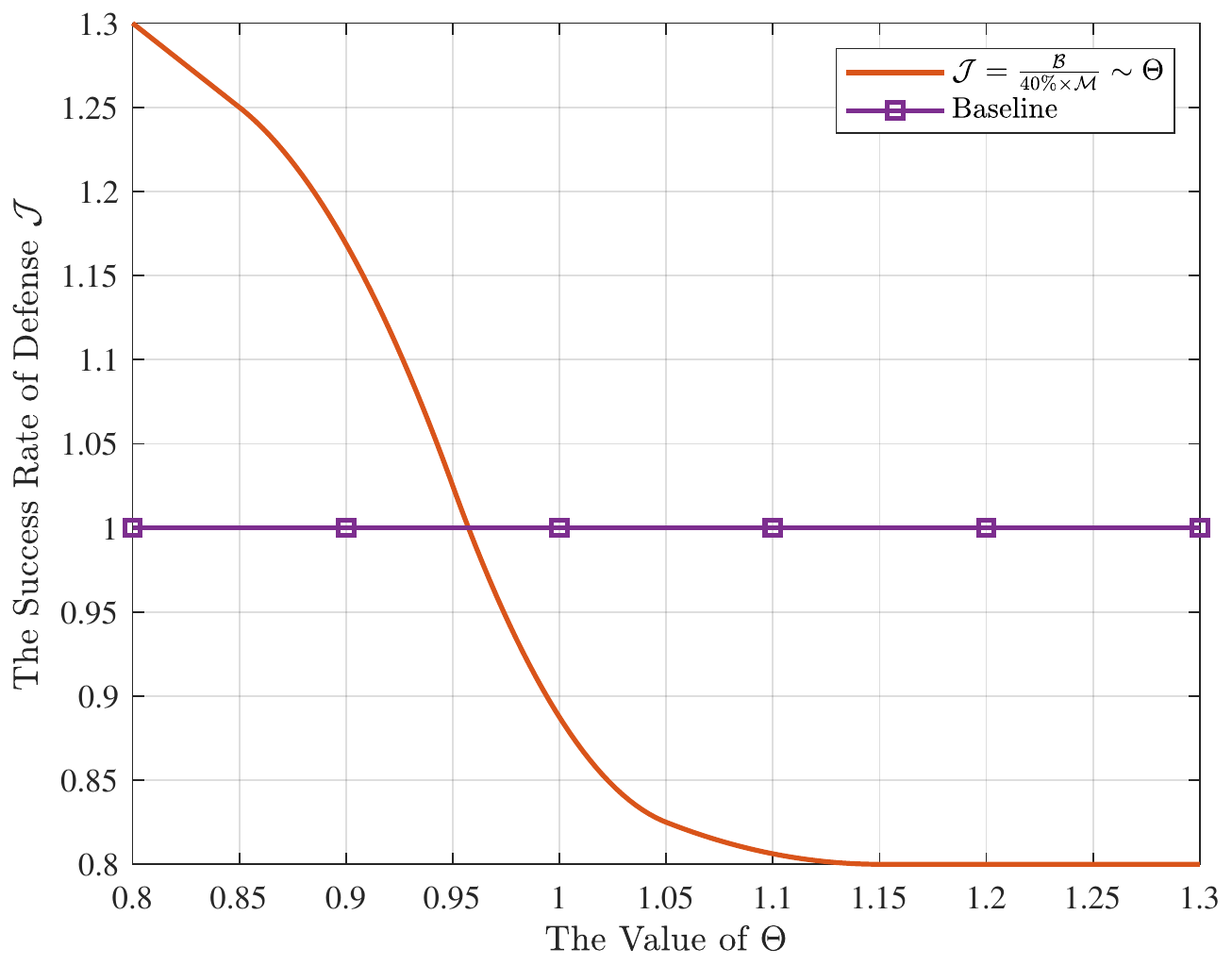}
	\caption{The Success Rate $\mathcal{J}$ with Different $\Theta$.}
	\label{fig:figure2label}
\end{figure}
To show the performance of KLAD, we give the success rate of defense based on different $\gamma$, such as in Table III. We know that the KLAD algorithm could effectively achieve defense performed better than existing methods without the test dataset in Table III. Besides, the success rate of defense is also different under different data poisoning ratios. Therefore, the higher the poisoning rate, the more conducive to the defense's success.
\begin{table}[!htb]
    \setlength{\tabcolsep}{1.6mm}{
	\begin{tabular}{lcccc}
		\hline
		Poisoning Rate $\gamma$ & $\gamma = 10\%$ & $\gamma = 20\%$ & $\gamma = 30\%$ & $\gamma = 40\%$ \\ \hline
		Poisoning VCs       & 5               & 10              & 15              & 20              \\ 
		Remove Malicious VCs     & 2               & 6              & 12              & 17              \\
		Remove Benign VCs		&10					&4					&2					&1				\\
		Defense Success Rate         & 40\%           & 60\%           & 80\%           & 85\%           \\ \hline
	\end{tabular}}
\caption{Performance of KLAD under $\gamma$ Data Poisoning Attacks}
\label{tab3}
\end{table}

Therefore, all the local models submitted are detected to send the benign local models to MAS and generate the global model in MAS. Finally, the benign global model is sent to each benign VCs.
\section{Conclusion}
In this paper, we mainly achieve the defense against label-flipping attacks in FL. Specifically, we construct a novel framework to add a component of data poisoning defense to the traditional FL architecture. In particular, we adopt Paillier cryptography to ensure the security of transmission of local model parameters. More importantly, regarding whether there is a test dataset for the submitted local models, we design  LMTV and KLAD defense algorithms to improve FL devices' safety. Finally, we demonstrated the effectiveness of the LMTV and KLAD through some experiments evaluation on the real-world datasets. In the future, we will continue researching the data poisoning defense algorithms to optimize the model's structure and improve defense performance without the trusted test dataset.


\begin{IEEEbiography}[{\includegraphics[width=1in,height=1.25in,clip,keepaspectratio]{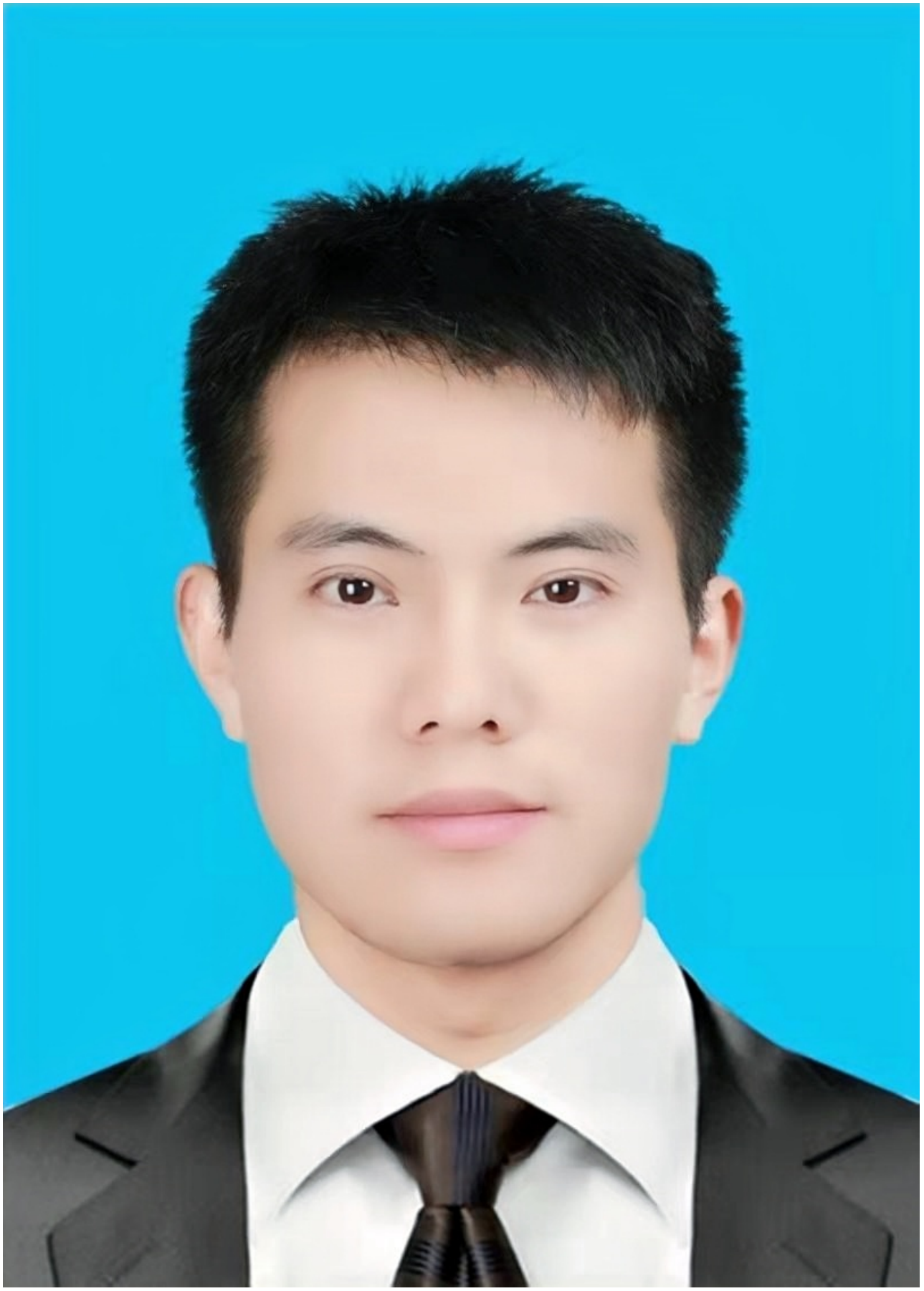}}]
	{Jiayin Li} received the B.Sc. degree in Electronic Information Engineering from Liaocheng University, Liaocheng, China, in 2015, and the MS degree in Communication Engineering, and the Ph.D. degree in Computer Science from Fuzhou University, Fuzhou, China, in 2018, 2022. Now, he is a teacher in the College of Computer and Cyber Security, Fujian Normal University. His research interests are in the area of data security, smart city, and machine learning. He has published some papers on intelligent information processing and intelligent transportation system security, including papers in IEEE TVT, IEEE Globecom, Journal on Communications, MDPI Sensors, IEEE WCSP.
\end{IEEEbiography}
\vspace{-33pt}
\begin{IEEEbiography}[{\includegraphics[width=1in,height=1.25in,clip,keepaspectratio]{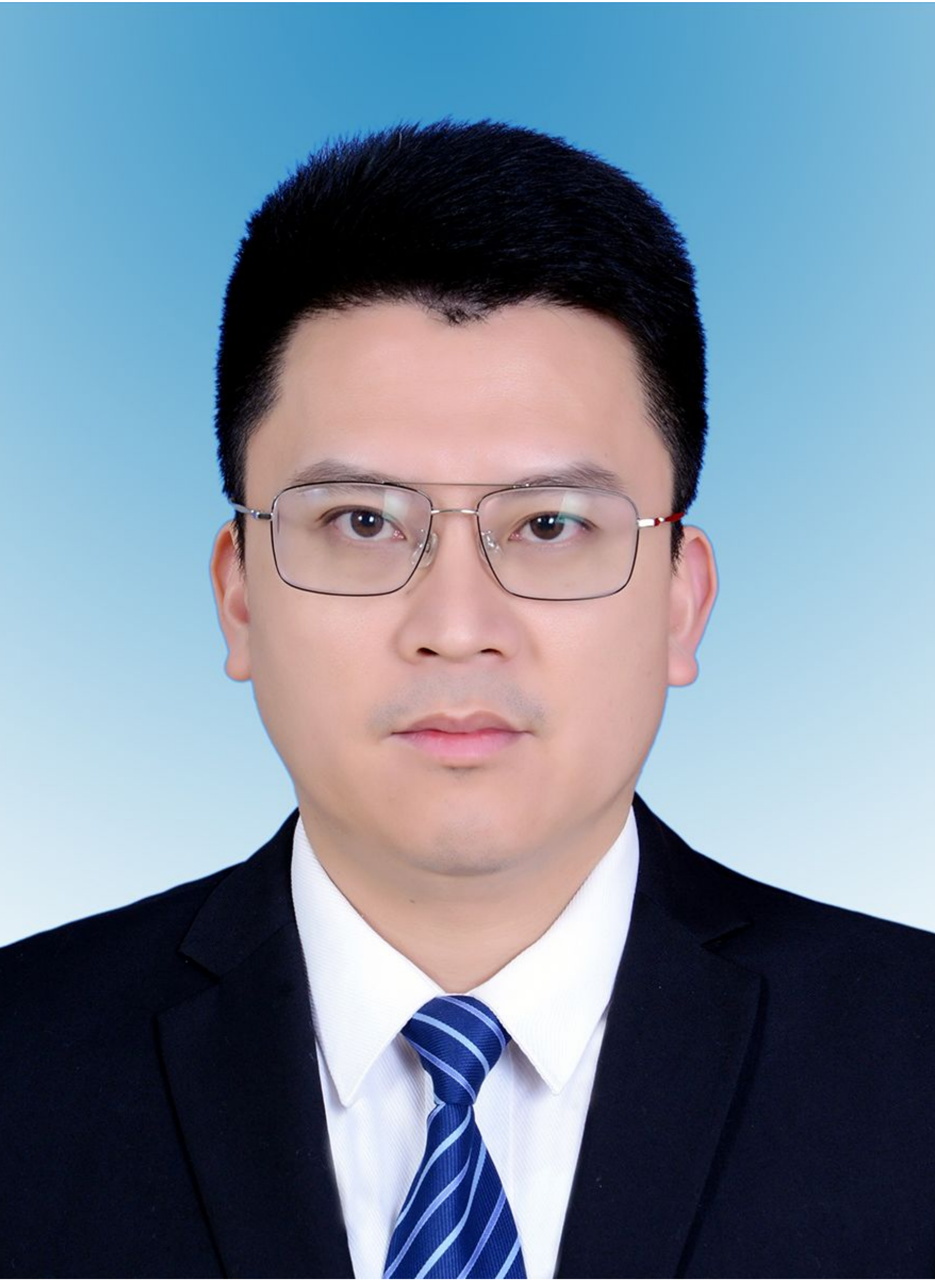}}]
{Wenzhong Guo} received the B.S.c and MS degrees in Computer Science, and the Ph.D. degree in Communication and Information System from Fuzhou University, Fuzhou, China, in 2000, 2003, and 2010, respectively. He is currently a full professor and dean in College of Mathematics and Computer Science at Fuzhou University. His research interests include intelligent information processing, sensor networks, network computing and network performance evaluation.
\end{IEEEbiography}
\vspace{-33pt}
\begin{IEEEbiography}[{\includegraphics[width=1in,height=1.25in,clip,keepaspectratio]{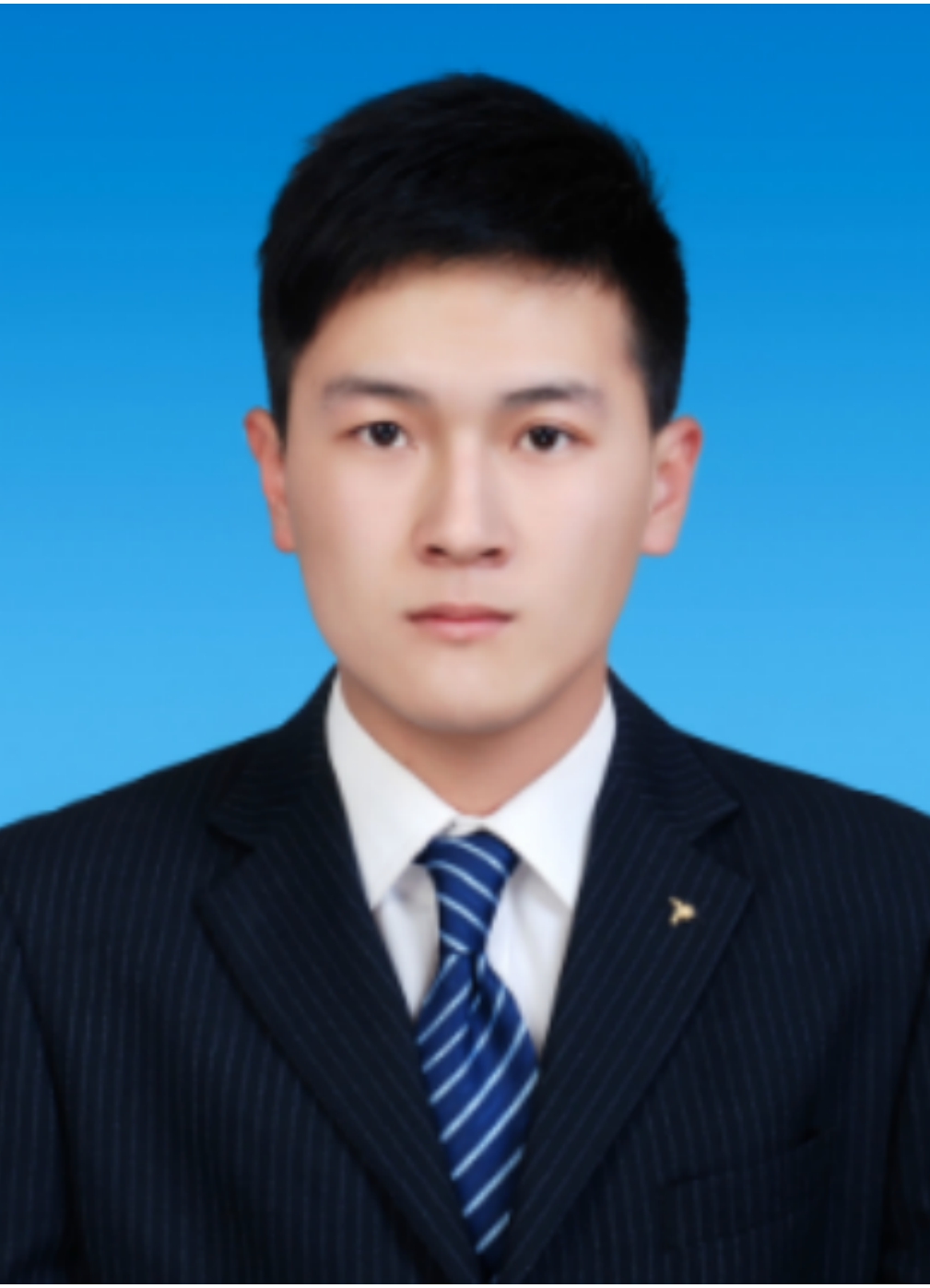}}]
{Xingshuo Han} received the B.Sc. degree in Communication Engineering from Liaocheng University, China, in 2014, and the MS degree in Circuit and System from Fuzhou University, Fuzhou, China, in 2018. Now he is PhD student in School of Computer Science and Engineering, Nanyang Technological University, Singapore. His research interests are in the area of machine learning, image processing and autonomous vehicles security.
\end{IEEEbiography}
\vspace{-33pt}
\begin{IEEEbiography}[{\includegraphics[width=1in,height=1.25in,clip,keepaspectratio]{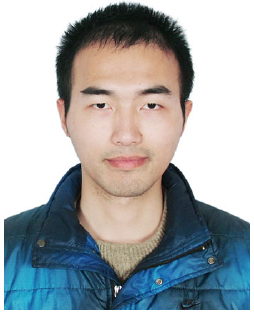}}]
{Jianping Cai} received the Master degree from Fuzhou University, Fuzhou, China, in 2016. He is pursuing his Ph.D. degree in College of Mathematics and Computer Science/College of Software, Fuzhou University, Fuzhou, China. His research interests include federal learning, Internet of things technology, machine learning, differential privacy and optimization theory.
\end{IEEEbiography}
\vspace{-33pt}
\begin{IEEEbiography}[{\includegraphics[width=1in,height=1.25in,clip,keepaspectratio]{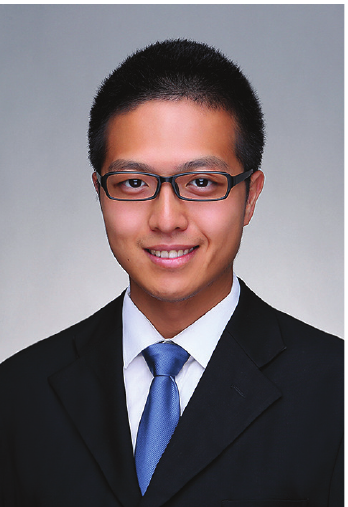}}]
{Ximeng Liu} (S'13-M'16) received the B.Sc. degree in electronic engineering from Xidian University, Xi'an, China, in 2010 and the Ph.D. degree in Cryptography from Xidian University, China, in 2015. Now he is the full professor in the College of Mathematics and Computer Science, Fuzhou University. Also, he was a research fellow at the School of Information System, Singapore Management University, Singapore. He has published more than 200 papers on the topics of cloud security and big data security including papers in IEEE TOC, IEEE TII, IEEE TDSC, IEEE TSC, IEEE IoT Journal, and so on. He awards ``Minjiang Scholars'' Distinguished Professor, ``Qishan Scholars'' in Fuzhou University, and ACM SIGSAC China Rising Star Award (2018). His research interests include cloud security, applied cryptography and big data security. He is a member of the IEEE, ACM, CCF.
\end{IEEEbiography}
 
\vfill

\end{document}